\documentclass[a4paper,12pt]{article}

\usepackage{amsmath,amsthm,amssymb,indentfirst,graphicx,caption,subcaption,setspace}
\usepackage[english]{babel}
\usepackage[big]{layaureo}
\usepackage{mathtools}
\usepackage{booktabs}
\theoremstyle{plain}
\usepackage{bbm}
\usepackage{appendix}
\newtheorem{proposition}{Proposition}
\newtheorem{remark}{Remark}
\newtheorem{lemma}{Lemma}

\newtheorem{definition}{Definition}

\usepackage[usenames]{color}
\usepackage{tikz}
\definecolor{red}{rgb}{1.0,0.0,0.0}

\definecolor{blu}{rgb}{0.0,0.0,1.0}

\definecolor{gre}{rgb}{0.03,0.50,0.03}

\definecolor{darkviolet}{rgb}{0.58, 0.0, 0.83}

\title{Habits and demand changes after COVID-19}
\author{Mauro Bambi\thanks{Corresponding author. Economics and Finance Department, Durham University, Durham, UK. email: mauro.bambi@durham.ac.uk}, \and Daria Ghilli,\thanks{Economics and Finance Department, LUISS Guido Carli University, Rome, IT.} \and Fausto Gozzi,\footnotemark[2] \and Marta Leocata\footnotemark[2]}
\date{February 2022}

\begin{document}
	\maketitle
	\begin{abstract}
In this paper, we investigate how a transitory lockdown of a sector of the economy may have changed our habits and, therefore, altered the goods' demand permanently. In a two-sector infinite horizon economy, we show that the demand of the goods produced by the sector closed during the lockdown could shrink or expand with respect to their pre-pandemic level depending on the lockdown's duration and the habits' strength. We also show that the end of a lockdown may be characterized by a price surge due to a combination of strong demand of both goods and rigidities in production.


\medskip

	
\end{abstract}

\newpage

\section{Introduction}

Habits have been largely recognized by the psychology and economics literature to influence significantly our consumption behavior. The way habits form and change over time depend among other things on the environment. For example, it is well-documented that people addicted to alcohol or other substances receive cues that trigger further abuse of these substances from the location where they consumed them in the past. Therefore, a change in the environment or context may alter significantly the habits either reinforcing or weakening them (e.g. Danner et al. \cite{Danneretal}).

The COVID-19 epidemics and, specifically, the social distancing and lockdowns have represented a drastic change of context for everybody. Being forced for long periods of time to stay at home and limit the physical interactions with other people have often been accompanied by changes in our consumption behavior. A large literature has already emerged about the effect of these restrictive measures on specific habits. For example, there are contributions on changes in eating/dietary habits and lifestyle during the lockdown (Dixit et al. \cite{Dixit}, Di Renzo et al. \cite{DiRenzoetal} and Sidor and Rzymski \cite{SidorRzymski}, among others). More generally, there is evidence that some habits have been reinforced, for example online shopping or using streaming services, while others weakened. Therefore, an open question is whether consumers will go back to their old habits such as shopping in the store or going to a cinema or the new habits will somehow replace the old ones (Sheth \cite{Sheth}).

From an economic perspective this would mean that a lockdown of a sector of production could change the consumers' habits and, therefore, alter their demand of goods so much so that the firms in the sector affected by the lockdown could find no longer profitable to remain active even after the end of the pandemic or, alternatively, they could have an incentive to expand their production to respond to a strong demand.

The existing literature on the macroeconomic consequences of a lockdown has investigated several interesting issues (e.g. Alvarez et al. \cite{Lippi}, Caulkins et al. \cite{Caukinsetal},  Giannitsarou et al. \cite{Giannitsaru}, Guerrieri et al. \cite{Guerrieri}).  Among these, Guerrieri et al. \cite{Guerrieri} is probably the closest in scope to our contribution as their objective is to show how and under which conditions a supply shock may lead to a demand-deficient recession. Similar to their investigation we focus on the effects of a lockdown on goods' demand in a multi-sector economy.

However, as far as we know, there is no contribution in the literature investigating how a change in the habits due to a lockdown may alter the consumption behavior after the pandemic. Could it be that the change in habits from old ones to new ones may lead an entire sector of the economy to disappear? If yes, how so? Could a government intervention avoid it? Could it be, instead, that the demand for goods not produced during the lockdown will expand after the pandemics? Could a change in the consumption composition push the good prices upward? Could prices rise above their pre-lockdown levels once the lockdown is over? The objective of this paper is to fill this gap in the literature and give an answer to these questions.

In order to do so, we study the demand and supply dynamics of a 2-sector infinite horizon economy. Each sector produces a differentiated good using a decreasing returns to scale technology with labor being the only input of production.  Sector 2 can be inactive either because a lockdown is imposed or because it is not profitable to produce. The latter may happen when the prevailing price of good 2 is too low to cover a fixed cost of production. Households inelastically supply an exogenously  given share of their labor endowment to each sector. Their labor income is then used to buy the two goods and to save through a foreign-asset in positive net-supply paying an exogenously given return. The foreign asset is denominated in term of good 1 and we also assume that good 2 is not tradeable
otherwise it could be bought from abroad even during a lockdown. In addition, sector 2 represents, in our mind, activities such as cinemas, restaurants  or stores whose services are hardly internationally tradeable.

As usual in habit formation models, households' utility is affected not only by consumption but also by habits. For reasons of clarity, we focus on a case where habits are formed only on good 1 consumption, although, as it will be clarified later, similar results can be obtained with habits formed on both goods, provided that some assumptions on the initial habits and on their dynamics are respected.\footnote{In the benchmark model, we also assume that there is no Inada condition on the marginal utility of the good not produced during the lockdown. However, the Inada condition can be easily restored as shown in Section \ref{Sec:provision} where we propose a variation of the benchmark model with a minimum provision of good 2 during the lockdown and zero fixed costs of production.}

Both sectors of production find profitable to be active before the pandemics. Then a temporary lockdown is imposed. As a result, Sector 2 shuts down and, consistently with the empirical evidence (e.g. Barrero et al. \cite{Barreroetal}), an exogenously given share of labor allocated to this sector is re-allocated to the other sector.  In the first model we are presenting, the length of the lockdown is unknown to the households; as such, the lockdown's end is perceived by them as an unanticipated shock.
An insight about our results can be obtained by looking at the effect of a lockdown on the good 2 market and, specifically, how this market looks like just before (BL) and just after the lockdown (AL or ALL, the latter meaning long lockdown). Figure \ref{figINTRO} illustrates the mechanisms driving our results.

\begin{center}
	\begin{figure}[h!]
		\centering 
		\begin{subfigure}[b]{0.7\textwidth}
			\includegraphics[width=\textwidth]{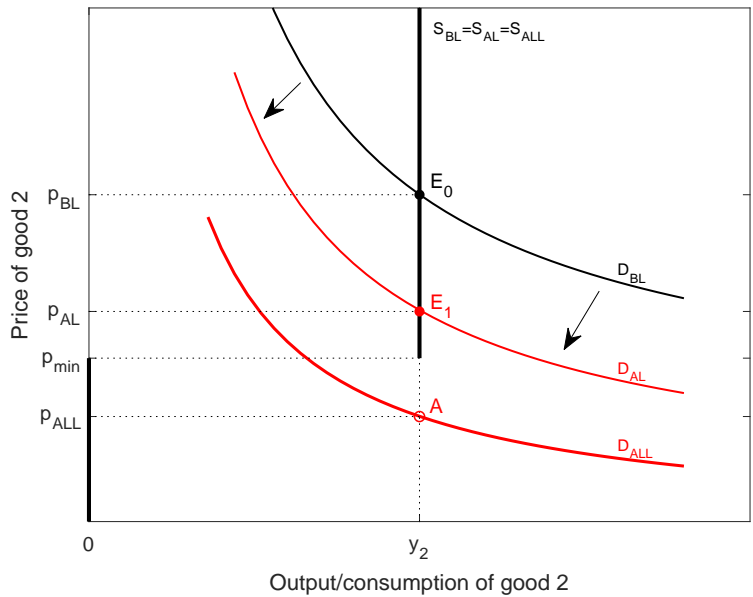}\caption{}
		\end{subfigure}
		\hfill	
		\begin{subfigure}[b]{0.7\textwidth}
			\includegraphics[width=\textwidth]{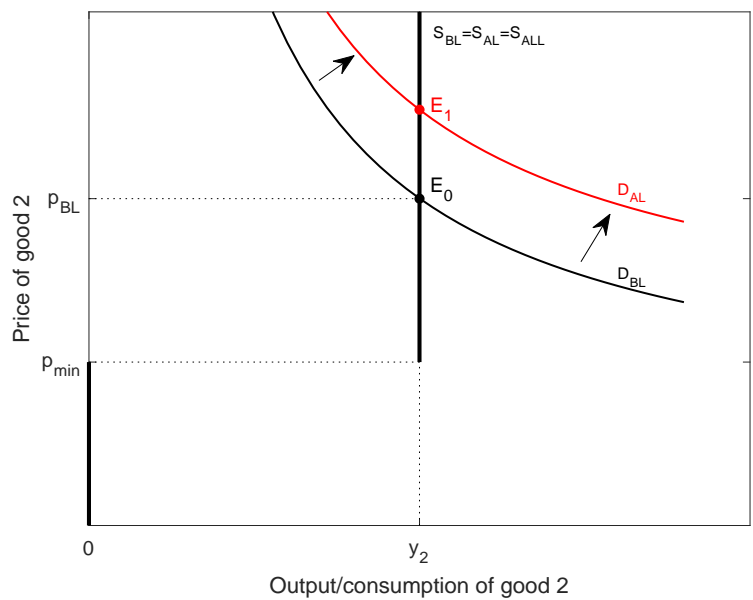}\caption{   }
		\end{subfigure}
		\caption{Lockdown effect on sector 2 depending on the relative strength of the satiation and substitability effect. (a) Satiation $>$ Substitutability. (b) Substitutability $>$ Satiation}\label{figINTRO}
	\end{figure}
\end{center}

In our benchmark model, the good 2 supply curve is vertical, since households inelastically supply an exogenously given share of their labor endowment to the sector, and it is the same before and after the lockdown.\footnote{This assumption implies a full labor readjustment immediately after the lockdown. In Section \ref{Sec:LaborChange} we propose a variation of the model to account for the change in the labor composition that we are seeing right now across developed countries.}  Given this simple pattern of labor allocation and the presence of a fixed cost in production, the supply curve is the truncated vertical thick black line $S_{BL}=S_{AL}=S_{ALL}$ drawn in both graphs of Figure \ref{figINTRO}.


On the other hand, the good 2's inverse demand is a downward sloping curve with the relative price of good 2 depending on the consumption of the two goods and the habits. Considering the case of substitutes goods (e.g. cinema and streaming services) any increase in the consumption of good 1 increases the price of good 2 shifting the inverse demand curve up (substitutability effect).

On the other hand, an increase in the habits formed on good 1 can have an ambivalent effect on the price of good 2. In particular, it may induce satiation or addiction on good 2 consumption. In the first case an increase in the habits shifts down the inverse demand curve (satiation effect) while in the second case the opposite happens (addiction effect).
Consider now the case that habits formed over good 1 consumption induce a satiation effect on good 2 consumption. The remaining question is which between the satiation effect and the substitutability effect prevails. Depending on the answer, the inverse demand curve after the lockdown, $D_{AL}$, will be higher or lower than before the lockdown, $D_{BL}$.

We prove that the magnitude of these effects depends on several factors. The substitutability effect may dominate or be dominated by the satiation effect depending on i) the habits speed of adjustment to changes in consumption, ii) the habits speed of convergence to their steady state value, iii) the cross-derivatives of utility between good 2 and good 1 consumption  and between good 2 consumption and the habits. The lockdown duration plays also an important role as the longer the lockdown will be, the larger will be the accumulation of the habits on good 1, and the stronger will be the satiation effect of the habits on good 2.

Consider now the following example to better understand the case when the satiation prevails. Suppose that during the lockdown we have reinforced an habit of binge-watching streaming movies/series (good 1); at the end of the lockdown we may be so satiated of watching movies/series that we will have a disincentive to go to the cinema (good 2). In particular, if the satiation effect dominates the substitutability effect then the inverse demand curve shifts down and the relative price of good 2, $p_{AL}$ or $p_{ALL}$ in Figure \ref{figINTRO}a, will be lower than before the lockdown, $p_{BL}$. Intuitively, this simply means that given the satiation effect we will accept to go to the cinema only if the ticket price is sufficiently low. Moreover, we prove that the price of good 2 will be lower the longer the length of the lockdown will be.

 If the lockdown is sufficiently long, the inverse demand curve could shift down so much that the relative price, $p_{ALL}$, of producing the quantity $y_2$ is below the threshold, $p_{min}$, for the firms in sector 2 to cover their fixed cost of production; for this reason sector 2 may remain inactive even after the lockdown since it is not profitable to do otherwise. Suppose, instead, that the lockdown is sufficiently short, then the inverse demand curve will shift down, for example from $D_{BL}$ to $D_{AL}$ in Figure \ref{figINTRO}a, but the relative price of good 2, $p_{AL}$, will be now higher then the price threshold, $p_{min}$, for firms in sector 2 to cover their fixed cost of production. In this case sector 2 will reopen after the lockdown with $E_1$ representing the new equilibrium.

On the other hand, if the substitutability effect dominates the satiation effect then the inverse demand curve shifts up and the new equilibrium at the end of the lockdown, $E_1$, will be characterized by a higher relative price of good 2. This case is described in Figure \ref{figINTRO}b. In this second case, we show that the presence of a sufficiently strong substitutability effect creates a pent-up demand of good 2 during the lockdown. This demand is higher than before the lockdown because the habits has accumulated faster during the lockdown with a positive effects on the demand of both goods. Once the lockdown is over, the pent-up demand will push the good 2 price up because now sector 2 is again open but rigidities in the supply side (i.e. $S_{AL}=S_{BL}$) prevent an expansion in production.\footnote{An even  stronger pent-up demand emerges if habits on good 1 induce addiction instead of satiation on good 2 consumption.} In addition, we show that also good 1 demand and consumption can be higher just after the lockdown implying a deterioration of the trade balance in the short run. As such, our model provides an explanation of the observed strong goods' demand observed across developed economies after the lockdown accompanied by an increase in goods prices.\footnote{This is also documented in the media. For example, Greg Ip has recently written in The Wall Street Journal that ``neither supply or demand by itself is increasing prices; it is an unusual combination of both''.}
	
	This mechanism provides a microfoundation for the emergence of a pent-up demand and its role in affecting the price dynamics after a lockdown induced recession. This explanation of the emergence of a pent-up demand during a lockdown clearly differs from other mechanisms found in the literature such as the one recently proposed by Beraja and Wolf \cite{BerajaWolf}, where the pent-up demand is related to the intertemporal substitution of different types of consumption (durables vs services) and, most importantly, it does not study the effect of a lockdown.

Figure \ref{figINTRO} shows a snapshot of the good 2 market just before and just after the lockdown. How will it change over time? If sector 2 will remain active just after the lockdown then over time the economy will converge back to its pre-pandemic configuration. On the other hand, if sector 2 will stay close even after the lockdown then it will remain close forever since over time the habits will increase even more and the satiation effect will become stronger and stronger. This outcome can be avoided with a policy intervention in favour of Sector 2 after the lockdown is lifted.\footnote{As it will be discussed later, the closure of Sector 2 is not necessarily Pareto optimal in our framework.} In particular, we show that any policy designed to provide incentives to the households to consume again good 2, for example through a system of vouchers as it was done by the UK government to convince people to go back to the restaurants after the first lockdown, and/or of subsidises in favour of the firms in sector 2 will help them to become profitable again. Interestingly, we show that such a policy should not be permanent but rather transitory. Intuitively, as the households will start to consume again good 2 they will reduce more and more their habits on good 1 till when the substitutability effect will dominate the satiation effect. From that point on, the government intervention becomes unnecessary.

Therefore, our model sheds light on the crucial role played by the habits in determining how the economy looks like after a lockdown with the habits characteristics affecting the direction of goods' demand changes. In fact, the same economy without habits would have not experienced any consequence after a lockdown with all the aggregate variables returning immediately to their pre-lockdown levels.\footnote{This economy without lockdown is described in the Supplementary Material.}

Several extensions are proposed. For example, in the benchmark model, labor readjusts immediately when the lockdown is over. This is contrary to what is now happening across developed economies which are experiencing large changes in the composition of labor forces, with many workers not coming back to their former jobs in the sectors closed during the lockdown. To account for this evidence, we consider a variation of the model where the change in labor composition during the lockdown is no more transitory but permanent. As a consequence, the supply curve in Figure 1 shifts to the left inducing a higher good 2's relative price. However, a change in the labor composition also affects the demand curve in interesting ways as production of good 1 will now expand. These effects on the demand curve adds to the the previously discussed substitutability-satiation effect. Conditions for an increase (decrease) of the relative price both in the short run and in the long run are investigated and compared with the results of the benchmark model. Interestingly, a change in labor composition may lead to a permanent higher (lower) good 2 price and net demand.

An insight about the consequence of announcing the duration of the lockdown is discussed in another extension of the main model. In this new framework, the households know with certainty the lockdown duration and, therefore, they can adjust their consumption behavior in advance. As a result, we find numerically that a long lockdown is less likely to lead after its end to the shutdown of a sector of production.

Moreover, we look at the case with the lockdown duration  being an exponential random variable and we show that
this case is somehow intermediate between the two above (unanticipated and anticipated end of the lockdown).


The last two extensions of the main model shows that the results discussed previously still hold when the economy is populated by two types of workers and when a minimum provision of good 2 is provided during the lockdown.

Finally the last section of the paper provides empirical evidence in support of the role of habits formed on one good in affecting the relative price of the other good. This is done by constructing a relative price index using the ratio of the CPIs of the admissions to movies, theatres and concertos, and of the cable and satellite TV services. Then we have used data from the American Time Use Survey (ATUS) on the time spent in TV watching to build a proxy for the habits. Various regressions show that habits are always significant and actually the main driver of the relative price movements.  

The paper is organised as it follows. In Section \ref{Sec:ModelSetup} we present the model setup. In Section \ref{Sec:2sectors} we describe the economy when two sectors of production are active while Section \ref{Sec:1sector} focuses on the case with only sector 1 being active. Section \ref{Section:LockdownQ}, uses the previous results to describe how the economy is affected by a lockdown with a focus on the habits and good 2 price dynamics. The main result is summarized in Proposition \ref{Prop:keyresult}. In this first model the end of the lockdown is unknown to the households and it is a unanticipated shock. A numerical example concludes this section and it helps to understand how a government policy could be designed to prevent the permanent shutdown of sector 2. In Section \ref{Sec:Extensions} we present several extensions to the benchmark model. Section \ref{CH:Empirical} provides empirical evidence in support of our theory. Finally, Section \ref{Sec:Conclusion} concludes the paper.  All the proofs are in Appendix.

\section{Model setup } \label{Sec:ModelSetup}

Consider an economy with two sectors producing two different final goods. Each sector has a unit mass of identical firms. Both goods are produced using a decreasing returns to scale production function $y=\ell^\alpha$ with labor the only input.

In both cases there is a fixed operating cost, $\tau$, denominated in units of the numeraire's output, which the firms pay to the households in each period to remain in the market. For example, $\tau$ could be a per period lump sum tax paid by each firm to the government for the permit to remain open and used by the government to pay a lump sum subsidy to the households.\footnote{As it will result clear later, the denomination of the fixed cost in unit of the numeraire's output is useful in this simple framework to have the threshold for the firm exiting the market in term of the relative price.} The profit maximization problem faced by a firm producing the final good $i$ writes:
$$\max_{\ell_i} \ \pi_i\equiv p_i \ell_i^\alpha-w_i \ell_i -\tau$$
with $i=1,2$. The final good 1 is the numeraire whose price, $p_1$, is normalized to one, while $p_2=p$ is endogenous and its value will be determined in equilibrium. As usual with decreasing returns to scale in production and a fixed cost, there is a threshold below which producing is no more profitable and the firm shuts down. In addition, we assume that sector 2 becomes inactive during a lockdown even if it would find profitable to produce. On the other hand, sector 1 remains active even during a lockdown if it is profitable to do so.

The economy admits also an infinitely-lived representative household whose preferences are represented by the utility function
$$\int_0^\infty e^{-\rho t} u(c_1,h,c_2) dt$$
where $h$ indicates the habits formed over the consumption of good 1 according to the equation
$$\dot {h}=\phi( c_1 - h)$$
with $\phi>0$ and $h(0)=h_0$ exogenously given.\footnote{The assumption of habits formed over only one good is introduced to make the analysis less cumbersome. Intuitively similar results should be obtained assuming that habits are formed over both goods with state equations having same functional form and the ratio of initial habits, $\frac{h_{10}}{h_{20}}$, and $\frac{\phi_1}{\phi_2}$  sufficiently large.} Concavity in the consumption of good $i$ implies that $u_{c_i}>0$ and $u_{c_ic_i}<0$. Habits can be harmful, $u_{h}<0$, or beneficial, $u_{h}>0$. We also assume joint concavity of $u(.)$ in the three variables $(c_1,c_2,h)$:
$$u_{hh}<0,  \qquad u_{c_1c_1}u_{hh}-(u_{ c_1h})^2>0, \qquad and \ \ |D^2u(c_1,c_2,h)|<0$$
where $D^2u(c_1,c_2,h)$ is the Hessian matrix of $u(.)$.\footnote{Observe that Bambi and Gozzi \cite{BambiGozzi} have shown that even without concavity the maximum principle may lead to an optimal and unique solution.}
In this first version of the model, we also assume that there is no Inada condition on the marginal utility of the two goods.\footnote{Although not a standard assumption, it is worth remembering that it has been used before, for example, in the structural change literature (e.g. Kongsamut et al. \cite{KRX}).} This allows the household to enjoy zero consumption and prevents scarcity from driving the relative price to infinity. Clearly, this is a useful assumption to model a lockdown with production of good 2 being zero. However, we will present in Section \ref{Sec:provision} a variation of the model where the Inada conditions are re-introduced and show that the main result of the model still holds.

\medskip

Each household inelastically supplies her labor to the firms; similar to Guerrieri et al. \cite{Guerrieri}, we assume that an exogenously given fraction $\xi$ of work time, $\bar\ell$, is supplied to sector 1 and the remaining $1-\xi$ to sector 2 when both the sectors of production are active; on the other hand, if one sector is inactive (e.g. lockdown) then an exogenously given constant share $a\in(0,1)$ of work allocated in the inactive sector is re-allocated to the active sector.\footnote{In Section 6, we describe a similar economy populated with two types of agents and we will show that similar results can be obtained in the case of a linear-quadratic utility function.} Households have also the ownership of the firms and, therefore, any profit is distributed back to them. They also receive a subsidy which in equilibrium is equal to the lump sum tax imposed to the firms. Finally, the households' budget constraint when both sectors are active is
$$\dot b+c_1+p c_2=rb+w_1\xi\bar\ell+w_2(1-\xi)\bar\ell +2 \tau +\pi_1+\pi_2 $$
where $b$ indicates the amount of foreign assets in positive-net supply and $r$ the constant (exogenous) world interest rate. On the other hand, if only sector 1 is active than the intertemporal budget constraint rewrites:
$$\dot b+c_1=rb+w_1[\xi+a(1-\xi)]\bar\ell + \tau +\pi_1 $$
Notice that in both cases we have indexed the bonds to good 1 consumption. As it will result clear later, this is done because  we will focus on the  case where good 1 is tradeable (e.g. streaming service) while good 2 is not (e.g. cinema). If you have a subscription to a streaming service such as Netflix then you can access it independently on your location. On the other hand, you cannot use a ticket of a cinema to see a movie in another cinema in another country.

Before investigating the economic effects of a lockdown, it is convenient to study two different problems. The first is an economy where both sectors are active and the other where only sector one is active.

\section{Economy with two active sectors of production}  \label{Sec:2sectors}

In this section, we consider the case of an economy with both sectors of production being active. This is, for example, the case for a sufficiently small fixed cost of production, $\tau$. In this framework, the households' optimization problem writes
$$\max_{c_1,c_2,b,h}\ \ \int_0^\infty e^{-\rho t} u(c_1,h,c_2) dt$$
subject to the following constraints
\begin{eqnarray}
	&& \dot b+c_1+p c_2=rb+w_1\xi\bar\ell+w_2(1-\xi)\bar\ell +2 \tau +\pi_1+\pi_2
	\label{b_state}\\
	&& \dot {h}=\phi( c_1 - h) \label{habits_state1} \\
	&& b(0)=b_0, \ \ \ h(0)=h_{0},  \ \ given
\end{eqnarray}
The inequality constraints, $h,c>0$, also hold. A given state-control quadruple $(c_1,c_2,h_1,b)$ is optimal if there exists absolutely continuous co-state functions $\mu$ and $\lambda$ such that
\begin{eqnarray}
	&& u_{c_1}+\mu\phi-\lambda=0 \label{foc:c1} \\
	&& u_{c_2}-p\lambda=0 \label{foc:c2}\\
	&& \dot\mu =(\phi+\rho)\mu-u_{h} \label{foc:h}\\
	&& \dot \lambda =(\rho-r)\lambda \label{foc:b} \\
	&& \lim_{t\rightarrow\infty}h\mu e^{-\rho t}=0 \label{tvc:1} \\
	&& \lim_{t\rightarrow\infty} b \lambda e^{-\rho t}=0 \label{tvc:2}
\end{eqnarray}
Observe also that $\lambda>0$ while the sign of $\mu$ depends on the habits being harmful or beneficial. In the first case, it is negative while in the latter it is positive.

On the other hand, the profit-maximization problem of firm $i$ leads to the following labor demand:
\begin{equation} \ell_i=\left\{\begin{array}{ll}
		\left(\frac{p_i\alpha }{w_i}\right)^{\frac1{1-\alpha}} & if \ \  p_i(1-\alpha)y_i\geq \tau \\ \\
		0 & \ \ otherwise
\end{array}\right. \label{Ffoc1} \end{equation}
therefore, there are two active sectors of production provided that the two sectors find profitable to produce. This happens when the price inequalities in (\ref{Ffoc1}) are respected and no lockdown is imposed.

%

The labor market clearing condition of sector 1 and sector 2 are respectively
\begin{equation}\begin{array}{ll}
	 \ell_1=\xi\bar\ell, & if \ \ \tau\leq (1-\alpha)y_1, \ \ and  \label{lmrktclear}\\
	 \ell_2=(1-\xi)\bar\ell, & if \ \ \tau \leq p(1-\alpha)y_2
\end{array}\end{equation}
On the other hand, the goods market clearing conditions are
\begin{equation}
	\begin{array}{ll}
	\dot b+c_1=rb+ y_1, & if \ \ \tau\leq (1-\alpha)y_1, \ \ and \label{cmrktclear}\\
	c_2=y_2, & if \ \ \tau \leq p(1-\alpha)y_2
\end{array}
\end{equation}
where $y_1=(\xi\bar\ell)^\alpha$ and $y_2=[(1-\xi)\bar\ell]^\alpha$.
Observe that if there is a positive-net supply of the bond, $b>0$, then the representative household will lend $b$ to expand her future consumption of good 1 from the amount produced within the country, $y_1$, to a maximum $y_1+rb$  by importing it. The opposite happens if there is a negative-net supply of the bond. In this case, the representative household may expand current consumption over the amount produced internally by borrowing from abroad but then she will repay this by a contraction of future consumption below $y_1$.

On the other hand, the final good 2 is assumed to be not tradable and, therefore, its consumption is always equal to the amount produced within the country, $y_2$.

\begin{definition}[Decentralized Equilibrium]
	A decentralized equilibrium of the economy is an allocation $(c_1,c_2,h_1,b,\ell_1,\ell_2)_{t\geq 0}$ and a price path $(w_1,w_2,p)_{t\geq 0}$ such that
	\begin{itemize}
		\item[i)] Given $(w_1,w_2,p)_{t\geq 0}$, the representative household chooses a quadruple $(c_1,c_2,h_1,b)_{t\geq 0}$ to maximize her intertemporal utility subject to (1)-(3).
		\item[ii)] Given $(w_i,p_i)_{t\geq 0}$, the representative firm in sector $i$ chooses $\ell_i$ to maximizes its profit subject to its production function, for $i=1,2$.
		\item[iii)] All markets clear in every period, i.e. (\ref{lmrktclear}) and (\ref{cmrktclear}) hold.
	\end{itemize}
\end{definition}

%

Observe that, at the decentralized equilibrium, we have that $u_{c_1}=u_{c_1}(c_1,h_1;y_2)$, $u_{c_2}=u_{c_2}(c_1,h_1;y_2)$, and $ u_{h}=u_{h}(c_1,h_1;y_2)$. This dimension reduction simplifies considerably our analysis and several results can be derived without choosing a specific utility function. To make the model even more tractable we will also assume from now on that $r=\rho$ and therefore $\lambda$ is constant. The only consequence of this assumption is that the economy will not grow over time.

\begin{proposition}\label{Prop:SS}
	A unique steady state exists with all the stationary variables function of the costate variable $\lambda$.
\end{proposition}

\medskip

It is worth noting that the existence and uniqueness of the steady state can be proved with and without the Inada conditions under some reasonable mild conditions on the parameters.\footnote{The interested reader may look at the proof in Appendix for further details.} We can now linearize (\ref{foc:c1})-(\ref{foc:h}), the habits equation (\ref{habits_state1}) and the final good 1 market clearing condition (\ref{cmrktclear}) around the steady state and eventually get the following result in the variables expressed as deviation from their steady state value, i.e. $\tilde x=x-x^*$.
\begin{proposition}\label{Prop:LocalDynamics}
	The local dynamics of the economy around its steady state is described by the following system of equations in the variables $(\tilde\mu,\tilde h, \tilde c_1,\tilde p, \tilde b,\lambda)$:
	\begin{eqnarray}
		&&\dot{\tilde\mu}=\left[\left(1+\frac{u^*_{c_1h}}{u^*_{c_1c_1}}\right)\phi+\rho\right]\tilde\mu+\frac{(u^*_{c_1h})^2-u^*_{c_1c_1}u^*_{hh}}{u^*_{c_1c_1}}\tilde h \label{eqeq:mu1}\\
		&& \dot{\tilde h}= -\frac{\phi^2}{u^*_{c_1c_1}} \tilde\mu -\phi\left(1+\frac{u^*_{c_1h}}{u^*_{c_1c_1}}\right)\tilde h \label{eqeq:h1}\\
		&& \dot{\tilde b}=r\tilde b- \tilde c_1 \label{eqeq:b1}\\
		&& \tilde c_1=-\frac{u^*_{c_1h}\tilde h +\phi \tilde \mu}{u^*_{c_1c_1}} \label{eqeq:c1}\\
		&& \tilde p=\frac{u^*_{c_2c_1}\tilde c_1+u^*_{c_2 h}\tilde h}\lambda \label{eqeq:p1}
	\end{eqnarray}
plus the transversality conditions.\footnote{Notice that, $\lambda$ will be determined using a TVC.}
\end{proposition}

\medskip

Equation (\ref{eqeq:p1}) is of particular importance in our framerwork as it will be used later to understand whether and under which conditions a lockdown may affect the prices so that the final good 2 could become more or less profitable to be produced  after the lockdown. At this stage of the analysis we can notice two things. First, that the two goods are substitutes if $$\frac{\partial p}{\partial c_1}=\frac{u^*_{c_2c_1}}{\lambda}>0 \qquad \Leftrightarrow \qquad u_{c_2c_1}>0$$ and complements otherwise. Second, and most importantly, the presence of the habits may reduce the price of commodity 2 $$\frac{\partial p}{\partial h}=\frac{u^*_{c_2h}}{\lambda}<0 \qquad \Leftrightarrow \qquad u_{c_2h}<0 $$ since $\lambda>0$. This condition means that the price of the final good 2 may decrease if the marginal utility of consuming that good decreases as the habits accumulate. In this case, although the habits are formed over the final good 1 consumption, they induce satiation in the consumption of good 2. The following result on the price change can be easily derived.

\begin{remark}
Assume that $u_{c_2c_1}>0$ and $u_{c_2h}<0$. Then
\begin{equation}dp<0 \qquad \Leftrightarrow \qquad \underbrace{u^*_{c_2c_1} d c_1}_{Substitutability \ Effect} <\underbrace{-u^*_{c_2h} d h}_{Satiation  \ Effect}. \label{KEYINSIGHT} \end{equation}
\end{remark}

\medskip

An example may help to understand this condition. Suppose that during a lockdown, the agents have reinforced a habits of binge-watching television programs. At the end of the lockdown, they may be so satiated of watching movies that they will accept to go to the cinema only if the ticket price is sufficiently low.

Remark 1 gives a first insight about the mechanism which may lead to an expansion or contraction of the demand of good 2 after a lockdown. This will be discussed later in Section 5 together with the role played by the length of a lockdown. Before doing that, it is useful to finish our analysis of the decentralized equilibrium for an economy with two active sectors and then with one active sector.

Let us proceed with our analysis and observe that a nice feature of our model is that (\ref{eqeq:mu1}) and (\ref{eqeq:h1}) is a linear system of ODEs in the variables $(\tilde \mu,\tilde h)$ which we can solve analytically.

\begin{proposition}\label{Prop:solsystem}
	Assume that  $u^*_{c_1h}<\bar u^*_{c_1h}$. Then the solution of the system of linear ODEs (\ref{eqeq:mu1}) and (\ref{eqeq:h1}) together with TVC (\ref{tvc:1}) exists and has the following form:
	\begin{eqnarray}
		&& \tilde h=\tilde h_0 e^{\psi_1 t} \label{eqeq:h}\\
	&& \tilde \mu=-\frac{\phi u^*_{c_1h}+(\phi+\psi_1)u^*_{c_1c_1}}{\phi^2}\tilde h \label{eqeq:mu}
	\end{eqnarray}
	where $\psi_1$ is the real and negative eigenvalue whose value depends on $\lambda$ (see Appendix A - Lemma 1).
\end{proposition}

\medskip

The inequality at the beginning of the proposition, $u^*_{c_1h}<\bar u^*_{c_1h}$, guarantees a positive and a negative eigenvalue and it is a necessary condition for the TVC (\ref{tvc:1}) to hold. The threshold value $\bar u^*_{c_1h}$ can be found in Appendix A - Lemma 1.
Once $\tilde h$ and $\tilde\mu$ have been found, we substitute them into (\ref{eqeq:c1}) to find how the dynamics of good 1 consumption and habits are related:
\begin{equation}\tilde c_1=\frac{\phi+\psi_1}\phi \tilde h. \label{eqeq:c1h} \end{equation}
Therefore, we have that the final good 1 is addictive if $\phi+\psi_1>0$ since its current consumption increases as the habits accumulate. On the other hand, good 1 is satiating when  $\psi_1+\phi<0$. Moreover, substituting $\tilde c_1$ into the relative price equation (\ref{eqeq:p1}) leads to
\begin{equation}\tilde p=\frac{(\phi+\psi_1)u^*_{c_2c_1}+\phi u^*_{c_2 h}}{\phi\lambda}\tilde h \label{eqeq:p}.\end{equation}
Therefore, a positive change of habits reduces the price of good 2 as long as \begin{equation} \underbrace{ u^*_{c_2c_1}\cdot \frac{\phi+\psi_1}\phi}_{Substitutability \ Effect}<\underbrace{-u^*_{c_2h}}_{Satiation \ Effect}. \label{KEY1}\end{equation}
Notice that this is the equilibrium counterpart of expression (\ref{KEYINSIGHT}) in Remark 1. The substitutability effect depends in equilibrium on the (steady state) cross-derivative of utility between good 2 and good 1 consumption, $u^*_{c_2c_1}$, the habits speed of adjustment to a change in good 1 consumption, $\phi$, and the habits speed of convergence to its steady state value, $\psi_1$. On the other hand the satiation effect depends on the (steady state) cross derivative of utility between good 2 consumption and habits, $u^*_{c_2h}$.

However,  the steady state value of our variables still depend on $\lambda$. To find it, we need to determine the solution of $b$, and in doing so we will find the value of $\lambda$ which makes the TVC (\ref{tvc:2}) hold.

\begin{proposition}\label{Prop:solsystem_b}
Assume that  $u^*_{c_1h}<\bar u^*_{c_1h}$. Then the dynamic path of $\tilde b$ is
\begin{equation}
	\tilde b=\frac{\phi+\psi_1}{\phi(r-\psi_1)}\tilde h \label{eqeq:b}
\end{equation}
with $\lambda$ equal to the value which makes the following equality hold
\begin{equation}\tilde b_0=\frac{\phi+\psi_1}{\phi(r-\psi_1)}\tilde h_0 \label{eqeq:lambda}\end{equation}	
\end{proposition}

\medskip

Without specifying an utility function it is not possible to find explicitly the value of $\lambda$. For this reason, we will choose,  in Section \ref{Section:LockdownQ}, a specific functional form  to assess the economic effects of a lockdown. Before doing so, we will briefly describe in the next section an economy where only sector 1 is active.

\section{Economy with one active sector of production}  \label{Sec:1sector}

In this section, we consider the case where the sector producing good 2 is inactive and a share $a\in(0,1)$ of labor, previously allocated in sector 2, is re-allocated in sector 1. This may happen either because the economy is in lockdown or because it is not profitable to keep sector 2 open. The households' optimization problem writes
$$\max_{c_1,h, b}\ \ \int_0^\infty e^{-\rho t} u(c_1,h) dt$$
subject to the following constraints
\begin{eqnarray}
	&& \dot b+c_1=rb+w_1[\xi+a(1-\xi)]\bar\ell + \tau +\pi_1 \label{onesector_bc}\\
	&& \dot {h}=\phi( c_1 - h)  \\
	&& b(0)=b_0, \ \ \ h(0)=h_{0},  \ \ given
\end{eqnarray}
The inequality constraints, $b,h,c>0$, also hold. The maximum principle leads to the FOCs (\ref{foc:c1}), (\ref{foc:h})-(\ref{tvc:2}).

The profit-maximization problem of the firms producing good 1 is also the same while the labor market clearing condition becomes
\begin{equation}
		\ell_1=[\xi+a(1-\xi)]\bar\ell, \qquad if \ \ \tau\leq (1-\alpha)y_1, \label{lmrktclearLD}
\end{equation}
where $y_1=\{[\xi+a(1-\xi)]\bar\ell\}^\alpha$ with $a\in(0,1)$, meaning that the final good 1 production has expanded since a share $a$ of labor previously allocated in sector 2 is now used in sector 1. The good market clearing condition is instead
\begin{equation}
			\dot b+c_1=rb+ y_1, \qquad  if \ \ \tau\leq (1-\alpha)y_1. \label{cmrktclearLD}\\
\end{equation}

Moreover, given the structure of our model, the functional form of the solution  with one or two active sectors of production is the same. In other words, Lemma \ref{Lemma:eigenvalues}, Proposition \ref{Prop:solsystem}, and Proposition \ref{Prop:solsystem_b} as well as equations (\ref{eqeq:h}),(\ref{eqeq:mu}),(\ref{eqeq:c1h}), (\ref{eqeq:b}), and (\ref{eqeq:lambda}) still hold in the case with only one sector. Of course, the path of the aggregate variables will be different since the absence of sector 2 and the labor reallocation from sector 1 to 2 will change the value of $\lambda$ and, therefore, the steady state values of the main aggregate variables; moreover, it will also affect the transitional dynamics since the eigenvalue $\psi_1$ depends on $\lambda$.  For the same reasons explained before, no further analysis is possible without assuming a specific utility function.

\section{Lockdown and its effects on the economy} \label{Section:LockdownQ}

In this section, we will study the effect of a lockdown on the economy. In particular, we will study its effect on the economy during the lockdown as well as after its end. To do so, we assume a linear-quadratic utility:\footnote{As already previously mentioned, the interested reader can find in Section 6 the case with two-types of workers and see that the same results emerge.}
\begin{equation}\label{eqn:uquadratic}
u(c_1,c_2,h)=a_{c_1}c_1 +a_{c_2}c_2+a_h h+\frac{a_{c_1c_1}}2c^2_1+\frac{a_{c_2c_2}}2c^2_2+\frac{a_{hh}}2h^2+a_{c_1c_2}c_1c_2+a_{c_1h}c_1h+a_{c_2h}c_2h,
\end{equation}
with the parameter conditions for concavity respected.\footnote{In the numerical example we will check that these conditions are indeed respected in the different scenarios we will investigate.} This functional form has been extensively used in the rational addiction literature (see among others Becker and Murphy \cite{Becker}, Dockner and Feichtinger \cite{Dockner}, and Iannaccone \cite{Iannaccone}).
Moreover, this functional form has several advantages. First, it makes the model analytically tractable since it is possible to find the shadow price, $\lambda$. In fact, with this functional form, all the second derivatives of the utility function are  constant and, therefore, it is immediate to see that the eigenvalue, $\psi_1$, will be no more a function of the co-state variable $\lambda$ whose value can be found using Proposition \ref{Prop:solsystem_b}. Second, we will be able to study the global dynamics and not just the local dynamics of the economy since an optimal control problem with linear-quadratic objective and linear states equations leads to a linear system of ODEs describing the dynamics of the economy. Last but not least, it is possible to perform a welfare analysis without the usual approximation issues (e.g. Benigno and Woodford \cite{Benigno}).

Before describing the timing of the shocks to the economy it is useful to find explicitly the steady state of an economy with two active sectors of production. Remember that in this case the output in the two sectors is   $y_1=(\xi\bar\ell)^\alpha$ and $y_2=[(1-\xi)\bar\ell]^\alpha$.

\bigskip

\begin{proposition}[\textbf{Steady state with 2 active sectors}]\label{Prop:prelock}
Assume that $a_{c_1h}<\bar a_{c_1h}$, $a_{c_1}>\underline a_{c_1}$, $a_{c_2}\geq\underline a_{c_2}$, and $b_0\in(\max\{\underline b_0,0\},\bar b_0)$.
Then, the steady state values of the main aggregate variables are
	\begin{eqnarray}
		&& h^*=  \frac{\phi(\psi_1-r)}{(\phi+r)\psi_1}\left[rb_0+y_1+\frac{r(\phi+\psi_1)}{\phi(\psi_1-r)}h_0\right] \label{h*}\\
		&& c^*_1= h^*\\
		&& b^*= \frac{h^*-y_1}r \label{eq:bstar} \\
		&& p^*= \frac{a_{c_2}+a_{c_2c_2}y_2+a_{c_1c_2}c^*_1+a_{c_2h}h^*}{\lambda} \label{pstar}\\
		&& \lambda=m_0+m_1\left[rb_0+y_1+\frac{r(\phi+\psi_1)}{\phi(\psi_1-r)}h_0\right] \label{lambda_LQ}
	\end{eqnarray}
with $m_0= a_{c_1}-\underline a_{c_1}>0$, $m_1= \frac{\phi(\psi_1-r)}{(\phi+r)\psi_1}\cdot\frac{(\phi+\rho)a_{c_1c_1}+(\rho+2\phi)a_{c_1h}+\phi a_{hh}}{\phi+\rho}<0$
and the threshold values $\bar a_{c_1h}$, $a_{c_1}$, $\underline b_0$, $\bar b_0$, and $\underline a_{c_2}$ reported in Appendix.
\end{proposition}

\medskip

Observe that conditions $a_{c_1h}<\bar a_{c_1h}$, $a_{c_1}>\underline a_{c_1}$, and $b_0\in(\underline b_0,\bar b_0)$ are needed to have both $h^*$ and $\lambda$ strictly positive while  condition $a_{c_2}\geq\underline a_{c_2}$ guarantees that the steady state prices are high enough for the firms in sector 2 to find profitable to produce.

We can now describe how the lockdown is modelled. At $t=0$ an unanticipated temporary lockdown is imposed on sector 2 which would have  found profitable to be active otherwise. The duration of the lockdown, $\tilde t$, is unknown and the lifting of the lockdown is another unanticipated shock from the agents' perspective. Once the lockdown ends, sector 2 will reopen if it is profitable to do so.

We will now describe the dynamics of the economy in the different phases of the lockdown.

%

\subsection{Arrival of the lockdown}

Suppose that the lockdown on sector 2 is unanticipated, temporary, and implemented at $t=0$. Since it is unanticipated, the agents will re-optimize following the problem setup explained in Section \ref{Sec:1sector} with initial condition $h_0$ and $b_0$, and parameter restrictions such that the economy without the lockdown (NL) would have produced both the final goods. This means that the parameters have been chosen so that  the price dynamics in an economy without a lockdown would have been:
\begin{equation}
	p_{NL}=p^*+ \frac{(\phi+\psi_1)a_{c_2c_1}+\phi a_{c_2 h}}{\phi\left\{m_0+m_1\left[rb_0+y_1+\frac{r(\phi+\psi_1)}{\phi(\psi_1-r)}h_0\right]\right\}}(h_0-h^*_{NL})e^{\psi_1 t}\geq \frac{\tau}{(1-\alpha)y_2}
\end{equation}
with $h^*_{NL}=h^*$ as found in equation (\ref{h*}). We will also assume that $h_0<h^*_{NL}$, although the other cases are also tractable as shown in Appendix A.

Since the duration of the lockdown, $\tilde t$, is unknown and the ending of the lockdown is modelled as another unanticipated shock, then the representative agent solves the same problem described in Section \ref{Sec:1sector} and we can use the results found previously. In particular, we have that the dynamics during the lockdown, i.e. in $t\in[0,\tilde t]$, is described by the following equations:
	\begin{eqnarray}
		&& h_L=h^*_L+(h_{0}-h^*_L)e^{\psi_1 t} \label{eqeq:hL} \\
		&& c_{1,L}=-\frac{\psi_1}{\phi}h^*_L+\frac{\phi+\psi_1}{\phi}h_L \\
		&& b_L= -\frac{\psi_1(\phi+r)}{r\phi(r-\psi_1)}h^*_L-\frac{y_{1,L}}r+\frac{\phi+\psi_1}{\phi(r-\psi_1)}h_L  \label{eqeq:bL}
	\end{eqnarray}
where the lockdown steady state habit stock is
\begin{equation}
	h^*_L=  \frac{\phi(\psi_1-r)}{(\phi+r)\psi_1}\left[rb_0+y_{1,L}+\frac{r(\phi+\psi_1)}{\phi(\psi_1-r)}h_0\right] \label{eqeq:hstarLock}
\end{equation}
with $y_{1,L}=\{[\xi+a(1-\xi)]\bar\ell\}^\alpha$ since a share $a\in(0,1)$ of labor in sector 2 has been re-allocated in sector 1. Notice also that for the same initial conditions, $h^*_L>h^*_{NL}$ since $y_{1,L}>y_{1,NL}$.

\subsection{After the lockdown}

At $t=\tilde t$ the agents re-optimize since the re-opening of the economy is again an unanticipated shock. Notice that the initial conditions for this problem are now $h_{\tilde t}$ and $b_{\tilde t}$ which can be found looking at equations (\ref{eqeq:hL}) and (\ref{eqeq:bL}). In particular, we have that:
\begin{eqnarray}
	&& h_{\tilde t}=h^*_L+(h_{0}-h^*_L)e^{\psi_1 \tilde t} \label{eqeq:inihAL} \\
	&& b_{\tilde t}= -\frac{\psi_1(\phi+r)}{r\phi(r-\psi_1)}h^*_L-\frac{y_{1,L}}r+\frac{\phi+\psi_1}{\phi(r-\psi_1)}h_{\tilde t}  \label{eqeq:inibAL}
\end{eqnarray}

In order to understand whether sector 2 will find profitable after the lockdow (AL) to be active or not we proceed as it follows. We assume that it stays open and then we check whether at the prevailing equilibrium prices it is actually optimal to do so.\footnote{This actually means that the representative household in solving her maximization problem does not take into account ex ante the price inequality for sector two to remain active. For example, this may happen when the fixed cost of production is a private information of the firm. Of course this ``asymmetric information" does not lead to a market failure if the equilibrium price, $p$, found by solving the problem in Section \ref{Sec:2sectors} is over the threshold and sector two remains active. However, it may represent a market failure if the equilibrium price is low and sector two firms do not find  profitable to remain active. As a consequence, the shutdown of a sector with the implied unemployment is not necessary Pareto optimal.} Therefore, we start assuming that labor readjusts immediately after the lockdown so that $y_1=(\xi\bar\ell)^\alpha$ and  $y_2=[(1-\xi)\bar\ell]^\alpha$. This assumption will be relaxed later in Section \ref{Sec:LaborChange} where the case of a permanent large change in the composition of the labor forces, as being observed across developed countries, will be studied.

The equilibrium path of the main variables in the case of a full readjustment of the labor force  to its pre-lockdown level can be found by adapting the results of Section \ref{Sec:2sectors}.

In particular, the equilibrium path of the habit stock in $t\in[\tilde t,\infty]$ is
\begin{equation}
	h_{AL}=h^*_{AL}+(h_{\tilde t}-h^*_{AL})e^{\psi_1 (t-\tilde t)} \label{habitpath}
\end{equation}
or equivalently
\begin{equation}
	h_{AL}=h^*_{AL}+\left[\underbrace{h^*_L-h^*_{AL}}_{>0}+\underbrace{(h_{0}-h^*_L)}_{<0}e^{\psi_1 \tilde t}\right]e^{\psi_1 (t-\tilde t)} \label{eqeq:hAL}
\end{equation}
where $$h^*_{AL}=h^*_{NL}=h^*$$ as proved in Appendix B with $h^*$ as found in Proposition \ref{Prop:prelock}.  Then, it is immediate to see that $h^*_L>h^*_{AL}$ because  after the lockdown the labor market readjusts immediately to the pre-pandemic situation, i.e. $y_{1,AL}=(\xi\bar \ell)^{\alpha}=y_{1,NL}$. In addition, we have also that $h_0<h^*_L$ since we assumed before that $h_0<h^*_{NL}=h^*_{AL}$. It can also be proved (see Proposition \ref{Prop:prelock}) that a full readjustment of the labor composition to its pre-lockdown level will imply the same shadow prices  $$\lambda_{NL}=\lambda_{AL}=\lambda.$$ Taking all these considerations into account, we can now write the equilibrium price dynamics of final good 2 in $t\in[\tilde t,\infty]$:
\begin{equation}
	p_{AL}=p^*+ \frac{(\phi+\psi_1)a_{c_2c_1}+\phi a_{c_2 h}}{\phi\left\{m_0+m_1\left[rb_0+y_1+\frac{r(\phi+\psi_1)}{\phi(\psi_1-r)}h_0\right]\right\}}\left[\underbrace{h^*_L-h^*}_{>0}+\underbrace{(h_{0}-h^*_L)}_{<0}e^{\psi_1 \tilde t}\right]e^{\psi_1(t-\tilde t)}
	\end{equation}
We are now ready to prove the main result of this paper. The following proposition describes the price dynamics when we are comparing two economies which are exactly identical at $t=0$. However, the first economy never experiences a lockdown and parameters are chosen so that both sectors remain active, while the other economy experience a lockdwon of $\tilde t$ periods. What will it happen after the lockdown? Will sector 2 firms find always profitable to reopen or not? The next proposition provides a taxonomy of all possible scenarios.

\begin{proposition}\label{Prop:keyresult}

 Consider the price dynamics in an economy without a lockdown (NL) and in an economy with a $\tilde t$-period lockdown (AL):
	$$
	p_{t,NL}=p^*+\frac{(\phi+\psi_1)a_{c_2c_1}+\phi a_{c_2 h}}{\phi \lambda}(h_0-h^*)e^{\psi_1 t} \ \ with \ \ t\in[0,\infty].
	$$
		$$
	p_{t,AL}=p^*+ \frac{(\phi+\psi_1)a_{c_2c_1}+\phi a_{c_2 h}}{\phi\lambda}\left[h^*_L-h^*+(h_{0}-h^*_L)e^{\psi_1 \tilde t}\right]e^{\psi_1(t-\tilde t)} \ \ with \ \ t\in[\tilde t,\infty]
	$$
with $\lambda$, $p^*$, $h^*$ as found in Proposition \ref{Prop:prelock} and $h^*_L$ as found in equation (\ref{eqeq:hstarLock}) and with $h_0<h^*$.
Then the following results hold:
\begin{itemize}
	\item if the satiation dominates the substitutability effect, $a_{c_2h}<\bar a_{c_2h}$, and the lockdown is sufficiently long, $\tilde t>\underline{ \tilde t}$, then
\begin{equation} p_{ t,AL}<p^*<p_{t,NL} \qquad with \ \ t\in[\tilde t,\infty]; \label{shutdown}\end{equation}
\item if the satiation dominates the substitutability effect, $a_{c_2h}<\bar a_{c_2h}$, and the lockdown is sufficiently short, $\tilde t<\underline{\tilde t}$, then
\begin{equation} p^*<p_{ t,AL}<p_{t,NL} \qquad with \ \ t\in[\tilde t,\infty];\label{noshutdown1} \end{equation}
\item if instead the substitutability dominates the satiation effect, $a_{c_2h}>\bar a_{c_2h}$, and the lockdown is sufficiently long, $\tilde t>\underline{ \tilde t}$, then
\begin{equation}p_{ t,AL}>p^*>p_{t,NL} \qquad with \ \ t\in[\tilde t,\infty] \label{noshutdown2};\end{equation}
\item if the substitutability dominates the satiation effect, $a_{c_2h}>\bar a_{c_2h}$, and the lockdown is sufficiently short, $\tilde t<\underline{ \tilde t}$, then
\begin{equation} p^*>p_{ t,AL}>p_{t,NL} \qquad with \ \ t\in[\tilde t,\infty];\end{equation}
\end{itemize}
where
\begin{small}
$$\bar a_{c_2h}\equiv -\frac{\phi+\psi_1}{\phi} a_{c_2c_1}, \qquad and \qquad \underline{\tilde t}=\frac{\ln(h^*_L-h_0)-\ln(h^*_L-h^*)}{|\psi_1|}$$
\end{small}
Similar results can be obtained for $h_0>h^*$, see the proof in Appendix A for further details.
\end{proposition}

\medskip

Several interesting facts emerge from this proposition. First, the result in (\ref{shutdown}) shows how deeply a lockdown may disrupt the economic activity. Consider for example an economy without a lockdown where both sectors are active as the fixed cost $\tau$ has been chosen so that the price threshold for sector two is below $p^*$. Then the same economy but with a lockdown sufficiently long would have implied a price dynamics such that the prices after the lockdown will be low enough for the firms in sector 2 to find no more profitable to remain active.


Observe also that the price after the lockdown tends to increase over time and to converge to $p^*$ which is the same as at the pre-pandemic level meaning that under that price sector 2 would find profitable to be active. Therefore, a policy implication of our model is that a government should subsidize firms in sector 2 after a lockdown till when the price has increased enough for them to find profitable to stay open.

On the other hand, the results in (\ref{noshutdown1})-(\ref{noshutdown2}) show that the good 2 price at the end of the lockdown could be permanently higher than at its pre-lockdown level. As it will be explained in details in Section \ref{sec:pent-up}, the increase in price depends on the pent-up demand formed during the lockdown period.

Another consideration is about the role of the habits initial condition, $h_0$. The condition $h_0<h^*$ matters only for the position of $p^*$ with respect to $p_{NL}$ and $p_{AL}$ while it has no role on the relation between these last two prices which is completely driven by the satiation and substitutability effect. The interested reader may find in the proof of Proposition \ref{Prop:keyresult} a more general formulation which consider also the case $h_0>h^*$.

In the next two subsections, we will consider and discuss two scenarios. In the first we will show numerically how negatively the economic activity can be affected by a lockdown when the satiation dominates the substitutability effect. In the second, we will consider an opposite scenario where the substitutability effect dominates the satiation effect and show how the pent-up demand on good 2 may drive a strong economic recovery.

\subsection{Price dynamics and sector 2 shutdown}\label{subsec:numerics}
We want now to illustrate through a numerical example the price dynamics as predicted by relation (\ref{shutdown}) and (\ref{noshutdown1}). This part is also useful because we want to verify at least numerically  that the several inequality constraints introduced throughout our analysis actually hold.

Let us start with the production side of the economy. Assume that before the lockdown the two sectors receive an equal amount of work, $\xi=0.5$, and that the work endowment is normalized to one, $\bar\ell=1$. During the lockdown, thirty percent, $a=0.3$, of labor in sector 2 is re-allocated in sector 1 consistently with the finding of Barrero et al. \cite{Barreroetal}. The labor share is $\alpha=0.7$ implying a total production $y_1=y_2=0.6156$ if both the sectors are active and $y_1=0.7397$ if sector 2 is inactive.

Looking now at the household problem, we set $a_{c_1}=a_{c_2}=1$ and $a_{c_1c_1}=a_{c_2c_2}=-1$. We also consider the case of substitute goods with $a_{c_1c_2}=0.3$, and of harmful habits, $a_h=-0.5$ and $a_{hh}=-1$. A positive change in the habits increases the marginal utility of good 1 consumption as $a_{c_1 h}=0.6$, while it decreases the marginal utility of good 2 consumption as $a_{c_2 h}=-0.1$. The last two parameter choices guarantee respectively a negative eigenvalue, so that there is convergence to the steady state, and satiation stronger than the substitutability effect. Finally we need to set the habits speed of adjustment to change in consumption, $\phi$. The drastic change in habits documented in the previously mentioned literature (see the introduction) seems to suggest habits promptly adapting to the new lifestyle imposed by the lockdown. For this reason we set $\phi=0.15$. This value implies that the half-life with which habits adjust toward a permanent change in $c_1$ is slightly more than one year which is a bit lower than the two years suggested in Carroll et al. \cite{Carrolletal} but basically the same as the value proposed by the literature on the equity premium puzzle (see Abel \cite{Abel} among others).\footnote{Carroll et al. \cite{Carrolletal} at page 345 explain that the choice of this value depends on the context examined. If the emphasis is on economic growth then lower values of $\phi$ can be chosen.}

The remaining parameters are set as it follows $r=\rho=0.01$, $h_0=0.5$, $b_0=1$ and $\tau=0.482$. With this choice of the parameters we have that the final good 1 is addictive since $\frac{dc_1}{dh}=\frac{\phi+\psi_1}\phi=\frac{0.1-0.0868}{0.1}=1.8679>0$.

All inequalities for concavity in Proposition \ref{Prop:prelock} are respected with this choice of the parameters. In particular, the utility function is strictly concave since $a_{c_1c_1}a_{hh}-a_{c_1h}^2=0.64>0$ and $|D^2u(c_1,c_2,h)|=-0.576<0$ and we have that the thresholds are equal to $\bar a_{c_1h}=1$, $\underline a_{c_1}=0.3258$, $\bar a_{c_2h}=-0.0396$, $\underline b_0=-60.8747$, and $\bar b_0=26.20$  meaning that all the parameters have been chosen within their constraints.

This choice of the parameters describes the first two cases, relations (\ref{shutdown}) and (\ref{noshutdown1}) respectively, in Proposition \ref{Prop:keyresult}. First observe that the length of the lockdown above which  sector 2 remains inactive is $\underline{\tilde t}=8.0548$ since for any lockdown longer than that the implied prices after the lockdown would be too low for the sector to find profitable to re-open.

\begin{center}
	\begin{figure}[!h]
		\centering
		\begin{subfigure}[b]{0.7\textwidth}
			\includegraphics[width=\textwidth]{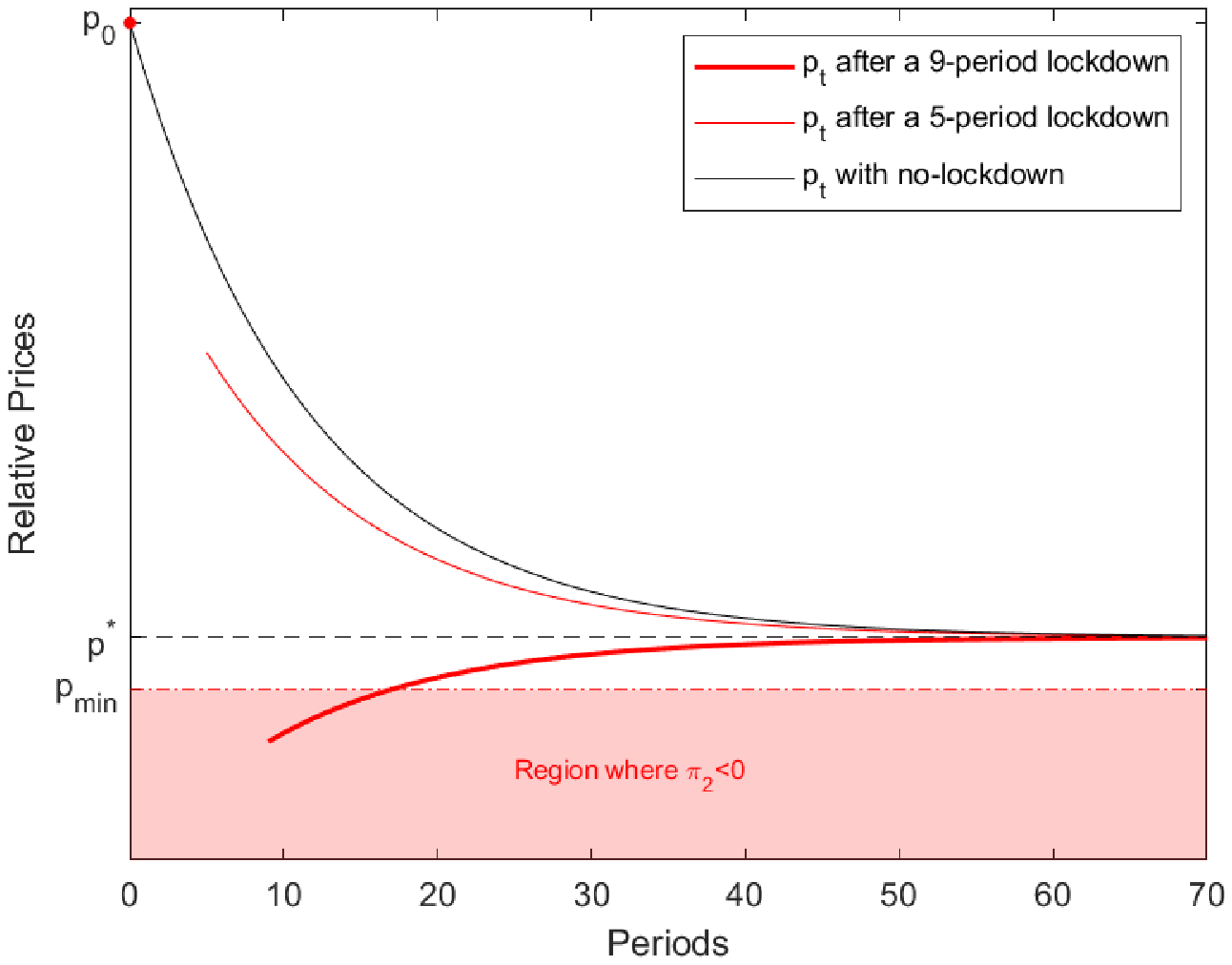}\caption{}
		\end{subfigure}
		\hfill	
		\begin{subfigure}[b]{0.7\textwidth}
			\includegraphics[width=\textwidth]{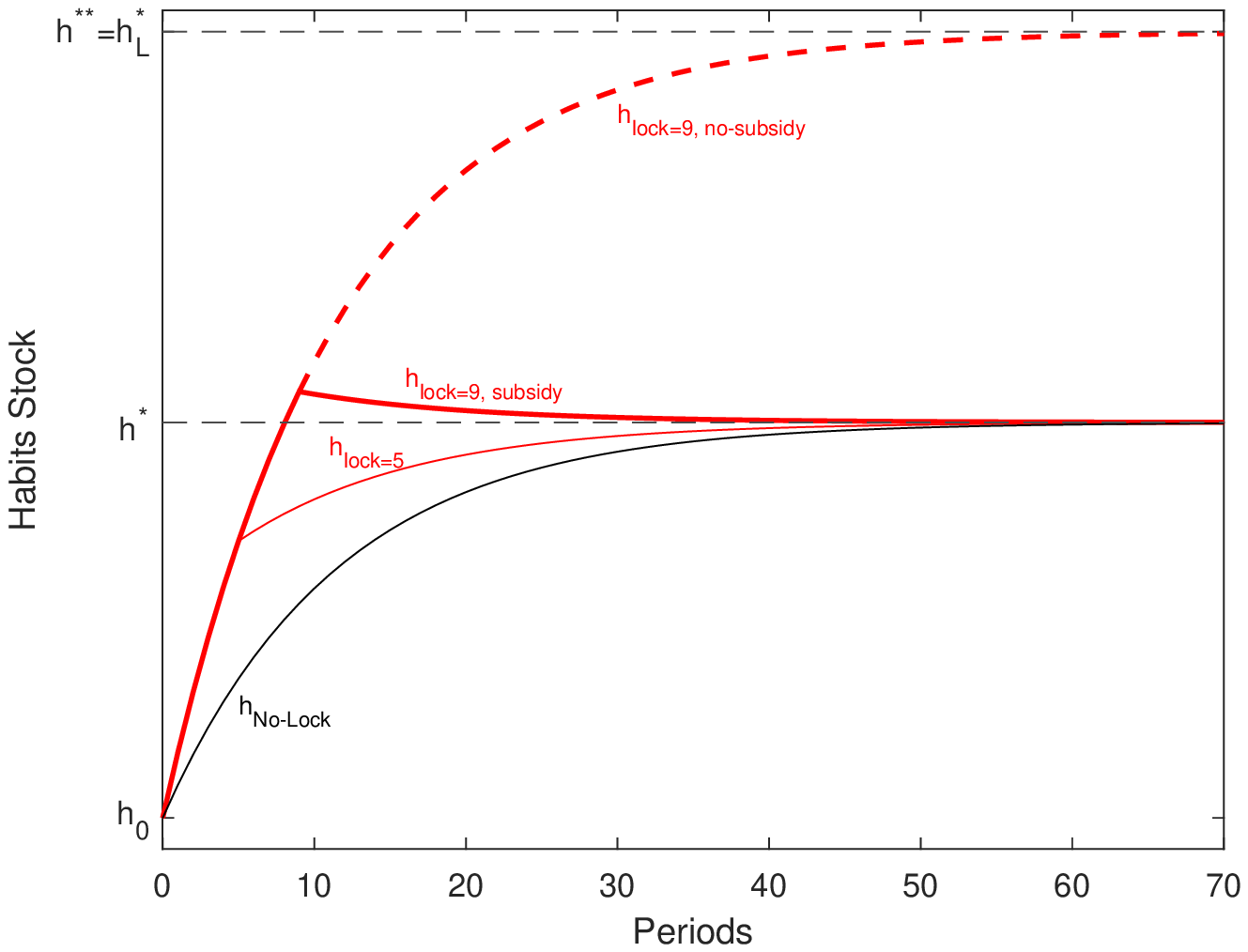}\subcaption{}
		\end{subfigure}
		\caption{Lockdown consequences when satiation dominates substitutability effect, $a_{c_2h}<\bar a_{c_2h}$: (a) Price dynamics. (b) Habits dynamics.}\label{Figure 1}
	\end{figure}
\end{center}

\vspace{-1cm}

This is illustrated in Figure \ref{Figure 1}a where the price dynamics is drawn assuming different lockdown durations. The red bold curve shows the price path when $\tilde t=9$ periods. Notice that the price of final good 2 at the time of the end of the lockdown is below the price threshold $p_{min}$ over which the sector does positive profit. Therefore, sector 2 will remain closed even after the lockdown since it is more profitable to do so. Observe however that if properly subsidize (i.e. a transitory  $\tau$ reduction) during the first 10 periods, then sector 2 would find convenient to remain open since eventually it will become profitable as the price path is no more in the red region. Finally the black and red line show the price path for different lengths of the lockdown both of them above $p^*$.\footnote{Notice that the relative price and habits path in the case of no-lockdown are qualitatively similar to those found empirically in Section \ref{CH:Empirical} for two specific sectors of the economy.}

On the other hand, Figure \ref{Figure 1}b shows the habits dynamics for different lengths of the lockdown. Notice that for the 9-period lockdown there are two paths. The first (red bold curve) shows the equilibrium habits dynamics assuming a government subsidy to sector 2 such as a reduction of the lump-sum tax so that the sector finds now profitable to remain active.\footnote{Alternatively the government could reduce the final good 2's price by introducing  vouchers as it was done by the UK government to subsidize the restaurants after the first lockdown. The scheme was called ``Eat out to help out'' and details can be found at the section ``Support for businesses and self-employed people during coronavirus'' of HM Revenue \& Customs webpage.}


 Coming back to Figure \ref{Figure 1}b, the other path (red dotted bold curve) illustrates the habits dynamics without government intervention; without it sector 2 will remain inactive even after the lockdown because the prevailing equilibrium price depicted in Figure \ref{Figure 1}a implies negative profits. Notice that the government intervention is crucial because without it the consumers will continue to consume only the final good 1 and, therefore, the habits will accumulate more and more and the relative price will become lower and lower.

\subsection{Price dynamics and pent-up demand}\label{sec:pent-up}

Let us now focus on the case where the substitutability effect is stronger than the satiation effect. As shown by relations (\ref{noshutdown1}) and (\ref{noshutdown2}) in Proposition \ref{Prop:keyresult}, this condition implies that the after-lockdown prices, $p_{t,AL}$, are higher than before the lockdown, $p_{t,BL}$. In this section, we will show through a numerical example how this positive change in prices is related to the good 2 pent-up demand formed during the lockdown when consumption expenditures on good 2 were not possible. For this purpose, let us consider an economy which is at its steady state when suddenly a lockdown is imposed.\footnote{The economy is at its steady state when the initial habits condition is $h_0=rb_0+y_1$.} This means that
	$$p_{t,BL}=p^*, \qquad and \qquad p_{t,AL}=p^*+SSE \cdot (h^*_L-h_0)(1-e^{\psi_1\tilde t})e^{\psi_1(t-\tilde t)}$$
	where $SSE\equiv \frac{(\phi+\psi_1)a_{c_2c_1}+\phi a_{c_2 h}}{\phi \lambda}>0$ since the substitutability effect dominates the satiation effect, $a_{c_2h}>\bar a_{c_2h}$.

Let us consider now the following numerical example. Suppose that the parameters describing the production functions are chosen as in the previous exercise. However, we change the initial condition of $b_0$ from $b_0=1$ to $b_0=\frac{h_0-y_1}r$ so that the economy starts at its steady state. Notice that, the resulting value of $b_0$ is negative meaning that the debt will be repaid by running a good 1 trade surplus, $TB^*=y_1-c_1^*>0$.  Moreover, we set $a_{c_2}=0.8$, $a_{c_2c_2}=-1.257$, $a_{c_1c_1}=-0.7$, $a_{c_1c_2}=0.6$, and $a_{c_2h}=-0.005$ in order to have i) a lockdown recession implying a -20\% deviation of GDP from its steady state level, ii) a substitutability effect stronger than the satiation effect (exactly the opposite case of the previous numerical example), iii) all the conditions for concavity satisfied, iv) all the inequalities in Proposition \ref{Prop:prelock} satisfied. In addition, we consider a monthly frequency and, therefore, we assume $r=\rho=0.001$ and a lockdown length of 9 months.

The drastic change in habits documented in the previously mentioned literature (see the introduction) seems to suggest habits promptly adapting to the new lifestyle imposed by the lockdown. For this reason we set $\phi=0.15$ consistently with the previous numerical exercise.




Using this choice of parameters we have computed the price dynamics before and after the lockdown. Figure \ref{pentup} shows the price deviations from its steady state level. The economy is characterized by a surge in $p$ after the lockdown; as it emerges from the figure, the price overshoots its steady state level $p^*$ by $1.6\%$ points at the date $t=9$ when the lockdown is lifted. This is  driven by the pent-up demand built up during the lockdown and the rigidities in production. As previously explained, the positive change in price happens when the substitutability dominates the satiation effect with the magnitude of the price adjustment depending crucially by the lockdown duration, the habits speed of adjustment to changes in consumption, the habits speed of convergence and the cross-derivatives of utility between good 2 and good 1 consumption  and between good 2 consumption and the habits.
	
To determine how large is the pent-up demand formed during the lockdown, we look at how much the good 2 demand has changed before and at the date of the reopening, $t=9$ . This is done in Figure \ref{pentup}(b) where we have drawn the inverse demand functions $p=f(c_2;h^*)$ and $p=f(c_2;h_{\tilde t,AL})$ with $f(p;h)=p=\frac{a_{c_2}+a_{c_2c_2}c_2+SSE\cdot h}\lambda$. Then, the change in good 2 demand at the fixed price $p^*$ is:
$$\Delta c^d_{2}(p^*)\equiv c^d_{2,AL}-c^{d*}_2=\frac{SSE}{-a_{c_2c_2}}(h_{AL}-h^*)$$
which, given our parameters' choice implies a 4.9\% increase in the good 2 demand. To meet this demand, output in sector two need to be expanded by an equal percentage. However this is not possible in our economy as the good 2 supply curve is vertical and good 2 cannot be purchased from abroad. As a consequence the expansion in demand fully translates in a price surge.

\begin{center}
	\begin{figure}[!h]
		\centering 
		\begin{subfigure}[b]{0.7\textwidth}
			\includegraphics[width=\textwidth]{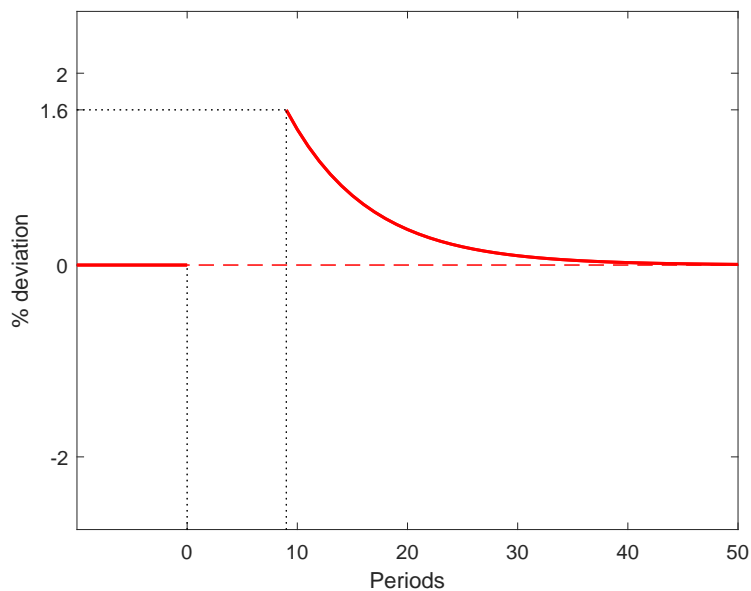}\caption{}
		\end{subfigure}
		\hfill	
		\begin{subfigure}[b]{0.7\textwidth}
			\includegraphics[width=\textwidth]{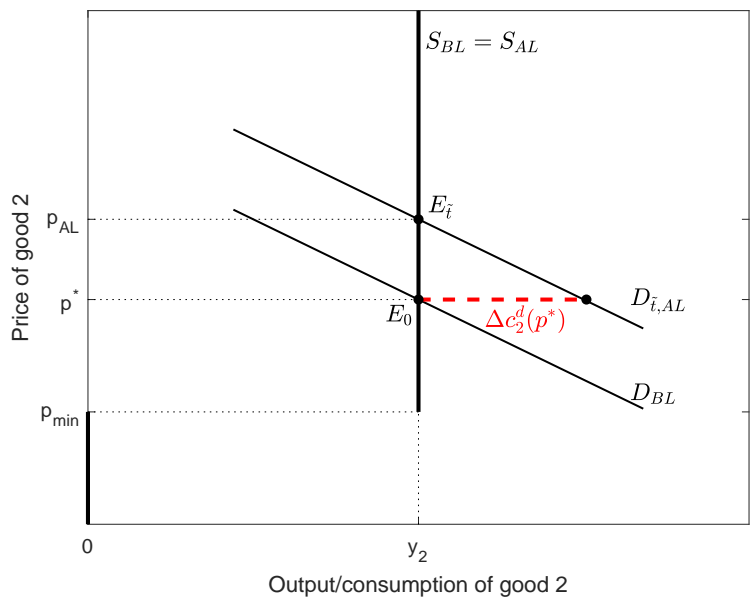}\caption{   }
		\end{subfigure}
		\caption{Pent-up Demand. (a) Relative price percentage deviation from its steady state $p^*$. (b) Good 2's demand change at the fixed price $p^*$.}\label{pentup}
	\end{figure}
\end{center}

\vspace{-1cm}

Interestingly, also the demand and consumption of good 1 is higher at the end of the lockdown. Based on our parameters' choice, we find that good 1 consumption has expanded by 5.45\% at date $t=9$. As a result of good 1 output expansion during the lockdown, also the debt at date $t=9$ is lower than at its steady state value. At the same time the trade balance, $TB_{\tilde t}\equiv y_1-c_{1,\tilde t}<y_1-c^*_{1}=TB^*$. The debt will be repaid by increasing the exportation of good 1 over time till converging to the pre-lockdown steady state. 	

This prediction of the model seems a plausible channel to explain the actual goods' demand and price surge experienced across developed countries after the lockdown.

%
%
%
%
%
%
%
%
\section{Extensions} \label{Sec:Extensions}
\subsection{Permanent change in labor composition}  \label{Sec:LaborChange}

In this section, we depart from the assumption that labor readjusts immediately after the lockdown. This variation of the model takes into account the large changes in the composition of labor forces currently observed across developed countries.
A change in the labor composition with more time now worked in sector 1 firms implies, in our model, a change $d\xi>0$. Since $y_1=(\xi\bar\ell)^\alpha$ and $y_2=[(1-\xi)\bar\ell]^\alpha$ then sector 1 will expand and sector 2 will shrink. A change in the labor composition will also affect the shadow price, $\lambda$, since in equilibrium it depends on the output of both sectors, see equation (\ref{lambda_LQ}).

For convenience, let us rewrite the equilibrium  relative price equation for the model with linear-quadratic utility
\begin{equation} p=\frac{\phi(a_{c_2}+a_{c_2c_2}y_2)+[(\psi_1+\phi)a_{c_1c_2}+\phi a_{c_2 h}]h}{\phi\lambda}.\label{price_equation}
\end{equation} 	
Using this equation we will study how the good 2 market looks like at the date of the re-opening, i.e. $t=\tilde t$. First, we observe that (see Appendix B):
\begin{equation}
dp=\underbrace{\left[\frac1\lambda \left(a_{c_2c_2}\frac{\partial y_2}{\partial\xi}-p\frac{\partial\lambda}{\partial \xi}\right)\right] }_{Labor \ Composition \ Multiplier \ (LCE)}\cdot \ d\xi+SSE\cdot dh \label{interpretation}
\end{equation}
where the first term on the right hand side is the labor composition effect while the last is the substitutability-satiation effect. We have also indicated with $SSE=\frac{(\psi_1+\phi)a_{c_1c_2}+\phi a_{c_2h}}{\phi\lambda}$, the multiplier of this last effect which is exactly the same as in the benchmark model. Differently from the benchmark case, the price dynamics is now also driven by the labor composition effect. Let us try to understand a bit more this new channel. Using the shadow price equation (\ref{lambda_LQ}) and the habits path (\ref{habitpath}), we can rewrite the multiplier of this new effect as follows:
\begin{equation}
	LCE=\frac1\lambda \left[\underbrace{a_{c_2c_2}\frac{\partial y_2}{\partial\xi}}_{>0}-p\left(a_{c_1c_2}+\frac{\phi}{\phi+\rho}a_{c_2h}\right)\underbrace{\frac{\partial y_2}{\partial\xi}}_{<0}\underbrace{-pm_1\frac{\partial y_1}{\partial \xi}}_{>0}\right] \label{interpretation1}
\end{equation}


A decrease in production of good 2, $\frac{\partial y_2}{\partial \xi}<0$, have the following effects. First, a shrink in production increases the price of that good by $a_{c_2c_2}\frac{\partial y_2}{\partial \xi}>0$ by shifting the supply curve to the left, from $S_{BL}$ to $S_{AL}$ in Figure \ref{figEXP1}. However, this effect on the equilibrium relative price can be mitigated or enhanced depending on the level of substitutability between the two goods and the role played by the habits.

In particular, the higher the level of substitutatibility between the two goods, term $pa_{c_1c_2}$ in the multiplier, the higher will be the positive effect on $p$  as an expansion in good 1 consumption tilts the inverse demand curve up. On the other hand, the higher the degree of satiation on good 2 implied by habits accumulation on good 1, term $-\frac{p\phi}{\phi+\rho}a_{c_2h}$ in the multiplier, the lower will be the effect on $p$ of an adjustment in production as now the inverse demand curve shifts down. Interestingly, the interaction between satiation and substitutability enters also in the labor composition multiplier and, in particular, we have that the former dominates the latter when $a_{c_2,h}<-\frac{\phi+\rho}{\phi}a_{c_2c_1}<-\frac{\psi_1+\phi}{\phi}a_{c_2c_1}\equiv\bar a_{c_2h}$ with $a_{c_2c_1}>0$. A strong satiation effect once compared with the substitutability, mitigates the positive change in the relative price due to the labor reallocation across sectors.

Moreover, an increase in production of good 1, $\frac{\partial y_1}{\partial \xi}>0$, has a positive effect, $-pm_1>0$, on $dp$ as $p$ is the inverse of its relative price. As a consequence it shifts the good 2's inverse demand curve up.
	
Clearly, the overall sign of the labor composition effect depends on the magnitude of these components. Interestingly, a sufficiently strong satiation effect, $a_{c_2,h}<<-\frac{\phi+\rho}{\phi}a_{c_2c_1}$, will imply not only $SSE<0$ but also $LCE<0$, and similarly to Remark 1 the relative price $p_{\tilde t}$ shrinks; for the same parameters' choice we can also prove that $p^*$ shrinks due the labor composition change and that $p$ may converges monotonically from above or from below to its new steady state value, $p^*_{AL}<p^*_{BL}$ (see also Appendix B). Figure \ref{figEXP1}(a) shows how the inverse demand curve, $D$, and the supply curve, $S$, adjusts. In particular, we have drawn these curves before the lockdown, i.e. $t=0$, once the lockdown is over, i.e. $t=\tilde t$, as well as at their final position when the economy has reached its steady state, i.e. $t\rightarrow\infty$. In Figure \ref{figEXP1}(a), the demand curve shifts down from $D_{BL}$ to $D_{\tilde t,AL}$ because the satiation effect is sufficiently strong. Then, assuming that $p$ converges from above to its steady state level, the inverse demand curve will shift down even further to its position $D_{\infty,AL}$. Note that it could be instead that the convergence is from below and that the new steady state $p^*_{AL}$ would lie between $p_{AL}$ and $p^*_{BL}$.

\begin{center}
	\begin{figure}[!h]
		\centering 
		\begin{subfigure}[b]{0.7\textwidth}
			\includegraphics[width=\textwidth]{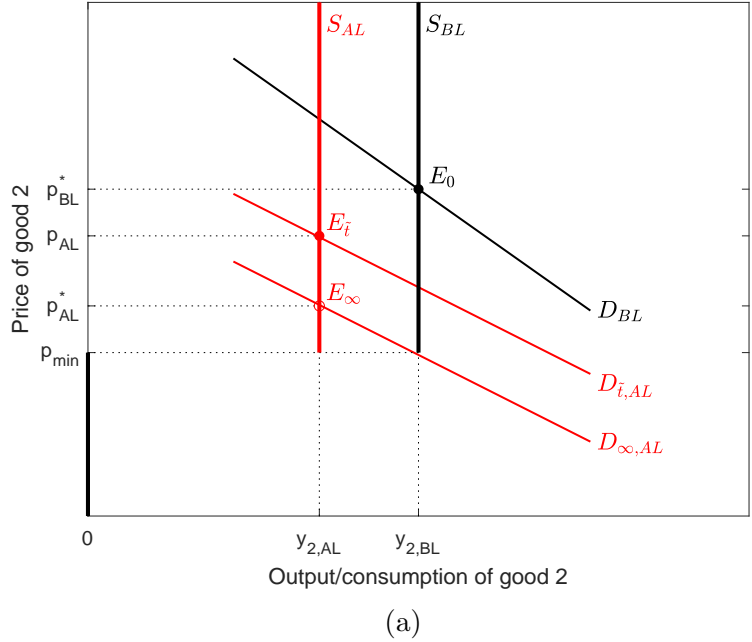}\caption{}
		\end{subfigure}
		\hfill	
		\begin{subfigure}[b]{0.7\textwidth}
			\includegraphics[width=\textwidth]{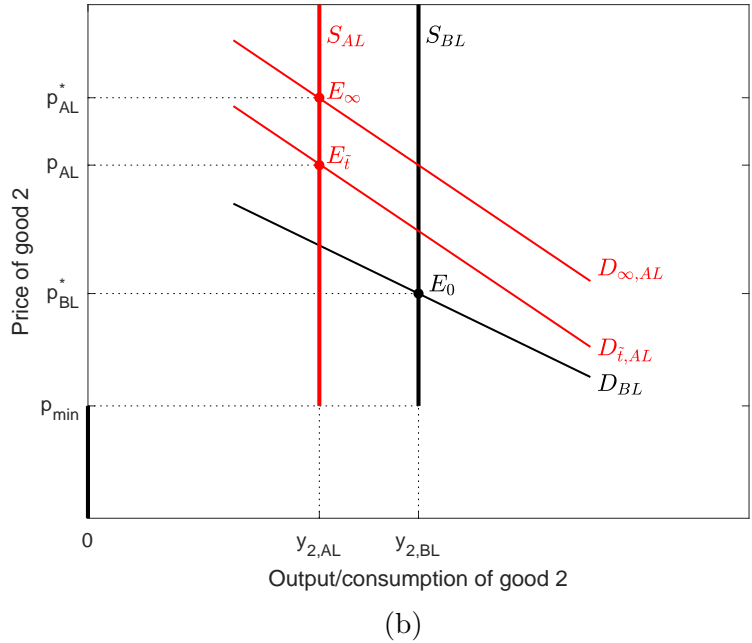}\caption{   }
		\end{subfigure}
		\caption{Lockdown effect on sector 2 in two cases. (a) Sufficiently Strong Satiation, $a_{c_2h}<<-\frac{\phi+\rho}{\phi}a_{c_2c_1}$. (b) Substitutability $>$ Satiation}\label{figEXP1}
	\end{figure}
\end{center}

\vspace{-1cm}
	
 On the other hand, if the substitutability dominates the satiation effect, $a_{c_2h}>\bar a_{c_2h}$, then, $LCE>0$,  $SSE>0$, and the demand curve shifts/tilts up from $D_{BL}$ to $D_{\tilde t, AL}$ while the supply curve shifts to the left from $S_{BL}$ to $S_{AL}$. As a result, $E_{\tilde t}$ will be the new equilibrium in the good 2 market once the lockdown is over. Such an equilibrium will be characterized by a higher relative price, $p_{\tilde t}$, than before the lockdown. In addition, also $p^*$ adjusts positively and the relative price converge to its new steady state value $p^*_{AL}>p^*_{BL}$. Observe that this is a sufficient condition meaning that the relative price may adjust positively even when the satiation effect weakly dominates the substitutability effect. Therefore, a readjustment in the labor composition reduces the range of parameters consistent with a permanent depression of sector 2. These considerations and conclusions are illustrated in Figure \ref{figEXP1}(b) where we shows the inverse demand curve and the supply curve adjustments before, and at the end of the lockdown $t=\tilde t$ as well as their final position when the economy has reached its steady state assumed to be above $p_{AL}$. The position of the steady state depends on several factors including the magnitude of the labor composition change.

In this last case, a labour composition adjustment affect the steady state prices, and in particular, we have that  $p^*_{AL}>p^*_{BL}$. From Figure 2(b), it is also clear that the raise in the price after the lockdown is due to a combination of a pent-up demand formed during the lockdown together with a shift of the supply curve to the left.

Moreover, it is interesting to have an insight about what happens to the price dynamics when we are not in these two extreme cases, for example, when we choose the parameters so that we have a mild satiation effect. As shown in Proposition \ref{confrpnps} in Appendix B, the economy may face the following price dynamics: first, at the end of the lockdown, the relative price could be lower than at its pre-lockdown level, $p_{AL}<p^*_{BL}$, however, the price may increase after the lockdown and  eventually converges to its steady state which is higher than before the lockdown,  $p^*_{AL}>p^*_{BL}$.

\subsection{Anticipated lockdown duration}\label{subsec:twosteps}
We are now interested in  studying an economy with a lockdown lasting $T>0$ periods with $T$ known by the consumers. To this purpose, the representative agent solves the following infinite horizon problem
\[
v(0,b,h)=\max_{c_1, c_2}\int_0^{\infty}e^{-\rho s}u(c_1, h, c_2) ds
\]
where $u$ has the quadratic form  \eqref{eqn:uquadratic}. As in the benchmark model, good 1 is always available while  good 2 is not available and therefore cannot be consumed during a lockdown. By the Dynamic Programming principle, we have that
\begin{equation}\label{DPP_0}
  v(0,b,h)=\max_{c_1, c_2}\left \{\int_0^{ T}e^{-\rho s}u(c_1, h, c_2)\, ds+v(T,b(T), h(T))\right\}.
\end{equation}

Moreover, it can be proved that, for any strictly positive values of $b$ and $h$, we have $v(T,b, h)=\tilde{v}(T,b,h),$
where $\tilde{v}(T,b,h)$ is the value function associated to the case of a strictly positive consumption of both goods. Therefore, we can write that
\begin{equation}\label{DPP_1}
  v(0,b,h)=\max_{c_1, c_2}\left \{\int_0^{ T}e^{-\rho s}u(c_1, h, c_2)\, ds+\tilde{v}(T,b(T), h(T))\right\}.
\end{equation}

This formulation of the problem allows us to split the original optimization problem into two sub-problems (the so called two-stage problem, see e.g. \cite{Tomiyama85}). We can then solve the problem as it follows. First we solve the infinite horizon optimization problem whose value function is $\tilde v(T,b(T), h(T))$, i.e.
\begin{equation}\label{eq:first_step_problem}
\tilde{v}(T,b(T), h(T))=\max_{c_1, c_2}\int_T^{\infty}e^{-\rho s}u(c_1, h, c_2) ds
\end{equation}
Crucially, the couple $(b(T), h(T))$ is treated here as given and it will be determined through the second sub-problem.

Once the first sub-problem has been solved, and the value function has been explicitly calculated, we can solve the second sub-problem \eqref{DPP_1}.
In Appendix B, we give a sketch of the solution of the two-stage problem.

We now introduce the index $TS,AL$, to indicate all the variables related to the infinite horizon maximization problem \eqref{eq:first_step_problem}, describing the economy after the lockdown; on the other hand, we use the index $TS,L$, the variables related to the finite horizon problem \eqref{DPP_1}.

In the following Proposition, we find the price after the lockdown, $ p_{t, TS, AL}$, by solving the two stage problem. We omit the details of the proof, since it is done exactly as the proof of Proposition \ref{Prop:keyresult}.


\begin{proposition}\label{Prop:keyresult2}
Under  the assumptions of Section \ref{Section:LockdownQ}, we have
\begin{equation}\label{eq:pTSAL}
p^*_{TS,AL}=\frac{a_{c_2}+a_{c_2c_2}y_2+(a_{c_1c_2}+a_{c_2h})h^*_{TS,AL}}{m_0+m_1\left[rb_{TS,L}(T)+y_1+\frac{r(\phi+\psi_1)}{\phi(\psi_1-r)}h_{TS,L}(T)\right]},
\end{equation}
and
	\begin{equation}\label{eqn:barp}
	p_{t, TS, AL}=p^*_{TS, AL}+ \frac{(\phi+\psi_1)a_{c_2c_1}+\phi a_{c_2 h}}{\phi\left\{m_0+m_1\left[rb_{TS,L}(T)+y_1+\frac{r(\phi+\psi_1)}{\phi(\psi_1-r)}h_{TS,L}(T)\right]\right\}}
	(h_{TS,L}(T)-h^*_{TS, AL})e^{\psi_1(t- T)}
	\end{equation}
	\end{proposition}
	
Comparing these price equations with their corresponding in Proposition \ref{Prop:keyresult} of Section \ref{Section:LockdownQ}, we notice that the only differences are that instead of $b_0, h_0$ one has now $b_{TS, L}(T), h_{TS, L}(T)$ and that the expression in rounded brackets multiplied by $e^{\psi_1 (t-T)}$ is replaced by $(h_{TS,AL}(T)-h_{TS,AL}^*)$. This observation follows from the fact that both $p_{t, TS, AL}$ and $p_{t,NL}$ come from the resolution of the same forward-backward system, but with different terminal conditions.

Differently from the case of an unanticipated lockdown length, the analytical derivation of the components inside these price equations is extremely cumbersome and any comparison between this case and the previous one with an unanticipated lockdown length needs to be done numerically.
	
Intuitively, we expect that the price in the case of an anticipated lockdown length should be higher than the price in the case the length is unanticipated, that is
	\begin{equation}\label{eqn:compprices}
	p_{t,AL}< p_{t,TS, AL}.
	\end{equation}
	where $p_{t, AL}$ is defined in Proposition \ref{Prop:keyresult}.
In fact, if the end of the lockdown is not anticipated, the households will behave as if the lockdown will last forever and, therefore, the adjustment of good 1 consumption will be strong as they want to compensate the reduction in utility due to the  unavailability of good 2 in $t\in(0,\infty)$. On the other hand, if the shock is anticipated, the households will still increase during the lockdown the consumption of good 1 but much less as they now need to compensate the reduction in utility due to the absence of good 2 in a shorter interval of time, $t\in[0,T]$. If consumption of good 1 increases less, then the habits will accumulate more slowly and, therefore, assuming that the satiation dominates the substitution effect, the price reduction will be lower than in the case of an unanticipated lockdown length.

Indeed, using the same numerical values of subsection \ref{subsec:numerics}, we find that  at $T=9$ it holds
$$
 p_{T,TS, AL}= 2.6330 >p_{\mbox{\footnotesize{min}}}=2.6100> 2.6070=p_{T, AL},
$$
and then sector $2$ will not shut down when the lockdown is anticipated and its length is $T=9$, differently from what happens when the lockdown is not anticipated (see subsection \ref{subsec:numerics}, Figure \ref{Figure 1} a)).

\subsection{Random lockdown duration}
We consider an economy where a lockdown has a random duration, $\tau$. More specifically, we assume that $\tau$ is an  exponentially distributed random variable, $\tau\sim \delta e^{-\delta s}$ with $\delta\geq 0$ and $s\geq 0$. The representative agent solves now the following intertemporal maximization problem

\[v(0,b,h)=\max_{c_1,c_2}\mathbb{E}\left[\int_0^\infty e^{-\rho s}u(c_1,h,c_2)ds\right]\]
the consumption of the first good, $c_1$ is always strictly positive, while the consumption of $c_2$ is zero during a lockdown, namely $c_2(t)\equiv 0$ for $t \in [0,\tau]$. By the Dynamic Programming principle, we have that
\begin{equation}\label{eqn:dynrandomprob}
v(0,b,h)=\max_{c_1,c_2}\mathbb{E}\left[\int_0^\tau e^{-\rho s}u(c_1,h,c_2)ds+v(\tau,b(\tau,h(\tau)))\right].
\end{equation}

\begin{lemma}\label{lem:DPP_1_random}
If $u$ is the quadratic function defined in \eqref{eqn:uquadratic}, and for any $b_0,h_0\in \mathbb{R}^+$ then
\begin{equation}\label{eqn:dynrandomproblem}
v(0,b, h)=\max_{c_1,c_2}\int_0^\infty e^{-(\rho +\delta)s}\left[u(c_1,h(s),c_2)ds+\delta \tilde v(0,b(s),h(s))\right]ds,
\end{equation}
where $\tilde{v}(0,b,h)$ is the value function associated to the case of strictly positive consumption of both goods.
\end{lemma}

By \eqref{eqn:dynrandomprob} and using similar arguments as in subsection \ref{subsec:twosteps}, Proposition \ref{Prop:keyresult2}, we can prove that for any  realization $\tau(\omega)$  of the random variable $\tau$, the price at the  reopening time $\tau(\omega)$ is a random variable and satisfies
	\begin{equation}\label{eqn:barp}
	p_{\tau(\omega)}=p^*(\tau(\omega))+ \frac{(\phi+\psi_1)a_{c_2c_1}+\phi a_{c_2 h}}{\phi\left\{m_0+m_1\left[rb_{TS,L}(\tau(\omega))+y_1+\frac{r(\phi+\psi_1)}{\phi(\psi_1-r)}h_{TS,L}(\tau(\omega))\right]\right\}}
	(h_{TS,L}(\tau(\omega))-h^*_{TS, AL})
	\end{equation}
	where
	$$
	p^*(\tau(\omega))=\frac{a_{c_2}+a_{c_2c_2}y_2+(a_{c_1c_2}+a_{c_2h})h^*_{TS,AL}}{m_0+m_1\left[rb_{TS,L}(\tau(\omega))+y_1+\frac{r(\phi+\psi_1)}{\phi(\psi_1-r)}h_{TS,L}(\tau(\omega))\right]},
$$
Observe that the price distribution can be derived since we know the distribution of $\tau$.
Moreover, a simple computation showing the  dependence of the relative price on the parameter $\delta$ can be done as follows. Given that $\tau$ is exponentially distributed, we have that
$$
\mathbb{E}[p(\tau)]=\int_0^{+\infty} \delta e^{-\delta s} p(s) ds.
$$
Therefore it follows immediately that
$$
\lim_{\delta \to 0} \mathbb{E}[p(\tau)]=0,
$$
meaning that for $\delta \to 0$ the problem converges  to the case of one active sectors of production in the benchmark model. On the other hand, we can also prove that
$$
\lim_{\delta \to +\infty} \mathbb{E}[p(\tau)]=\int_0^{+\infty} \delta_0 p(s)\, ds=p_0=p_{NL},
$$
where $\delta_0$ is the Dirac in zero. The last expression means that for $\delta \to + \infty$ the problem converges to the case of two active sectors of production in the benchmark model.

\subsection{Mimimum good provision under a lockdown}\label{Sec:provision}

Consider an economy very similar to the one discussed in Sections \ref{Sec:2sectors} and \ref{Sec:1sector} but with the following differences. First, there isn't a fixed cost of production and, therefore, the firms in both sectors of production find always profitable to produce. Second, a government intervention guarantees that during the lockdown a minimum provision of good 2 is produced and distributed to the households. The idea is that the consumers need an (exogenously given) subsistence level of good 2, $\underline{c}_2$, and the government acts to guarantee that during a lockdown.

More specifically, the government allows the firms in sector 2 to produce during the lockdown a quantity $y_2=\underline{c}_2$ using a share $\chi$ of labor that would be, otherwise, unused, $(1-a)(1-\xi)\bar\ell$. The workers will receive the same pre-lockdown wage, $w_{PL}$, and the government will buy the goods so produced at the market price charged just before the lockdown, $p_{PL}$, and distribute them to the households. The government expenditure will be financed by levying a lump sum tax, $T$, to the consumers. The government balanced budget constraint will be $$p_{PL}\underbrace{[\chi (1-a)(1-\xi)\bar\ell]^\alpha}_{\equiv \underline{c}_2}=T$$

The representative household budget constraint (\ref{onesector_bc}) will now become $$\dot b+c_1=rb+w_1[\xi+a(1-\xi)]\bar\ell + w_2\chi (1-a)(1-\xi)\bar\ell +\pi_1+\pi_2-T$$

while the instantaneous utility function becomes $u(c_1,\underline{c}_2,h)$ during a lockdown. Observe that the Inada conditions on the marginal utility of consumption can be now assumed both in the economy under and without a lockdown since we do not need to impose zero consumption any more.\footnote{A possible example of an utility function respecting the Inada condition is $$u(x)=\frac{x^{1-\theta}-1}{1-\theta} \ \ with \ \ \theta\neq 1$$
	where
	$$x=\left(\frac{c_1}{h^{\gamma}}\right)^{\eta_1}(c_2)^{\eta_2}.$$}

This variation of the model does not lead to any relevant difference with respect to the benchmark model as described in Section \ref{Sec:2sectors} and \ref{Sec:1sector}. In fact, all the results in these sections continue to hold with the only exception that firms in sector 2 will always remain active. This is because at the equilibrium the good 1 market clearing condition is actually the same in both versions; in fact, the lump-sum tax necessary to cover the government expenditure to purchase $\underline{c}_2$ is equal to the good 2 firms' revenues which are equal to the profits plus the production costs.

However, it is no more possible to derive analytically the results in Section \ref{Section:LockdownQ} because of the different specification of the instantaneous utility function. As previously explained, it is not possible to find the shadow price, $\lambda$, without a linear-quadratic utility specification. A quantitative investigation of this variation of the model is left for future research.

\subsection{Economy with two types of workers}

The economy is populated by a unit mass of  atomistic workers $n\in(0,\bar\ell)$ with $\bar \ell=1$. Assume also that the share of type-1 workers is $\xi$ and of type 2 workers is $1-\xi$, with $\xi\in(0,1)$ exogenously given. The two types of workers are inherently identical but the following labor allocation holds. Type-1 workers always work in sector 1 while type-2 workers work in sector 2 if this sector is active. If it is not then a share $a\in(0,1)$ of them moves to work in sector 1. Moreover type-$i$ workers have the ownership of sector-$i$ firms. Lockdown affects only sector 2 and therefore type-2 workers. Also $c_{ij}$ indicates $i$-type workers consumption of good $j$. Also both the agents form habits over good $1$ but not over good $2$.

The objective of this section is to show that the equilibrium price equation found in the model presented in Section 5, is actually the same in the case of two types of workers when the utility function is linear-quadratic. To this purpose, we just focus on the case of two active sectors of production and we leave further investigation for future research.

The optimization problem of a type-1 worker is
 $$\max_{c_{11},c_{12},b_{1}}\ \ \int_0^\infty e^{-\rho t} u(c_{11},c_{12},h_{1}) dt$$
 subject to the following constraints
 \begin{eqnarray}
 	&& \dot b_{1}+c_{11}+p c_{12}=rb_1+ \  w_1 +\frac{\tau+\pi_1}\xi\\
 	&& \dot {h}_1=\phi( c_{11} - h_1) \label{habits_state1} \\
 	&& b_1(0)=b_{10}, \ \ \ h_{1}(0)=h_{10},  \ \ given
 \end{eqnarray}
 where exactly as done before, we have normalized the price of good 1 to one and $p$ is the relative price of good 2, and we have indexed the bonds to good 1 consumption.

 \bigskip

 On the other hand, the optimization problem of a type-2 worker is
 $$\max_{c_{21},c_{22},b_{2}}\ \ \int_0^\infty e^{-\rho t} u(c_{21},c_{22},h_{2}) dt$$
 subject to the following constraints
 \begin{eqnarray}
 	&& \dot b_{2}+c_{21}+ p c_{22}=rb_2+ w_2+\frac{\tau+\pi_2}{1-\xi}\\
 	&& \dot {h}_2=\phi( c_{21} - h_2) \label{habits_state1} \\
 	&& b_2(0)=b_{20}, \ \ \ h_{2}(0)=h_{20},  \ \ given
 \end{eqnarray}

Observe that the first order conditions of the two agents are the same beside the index. In particular, a given state-control quadruple $(c_{i1},c_{i2},h_i,b_i)$ is optimal if there exists absolutely continuous costate functions $\mu_i$ and $\lambda_i$ such that

\begin{eqnarray}
	&& u_{i1}+\mu_i\phi-\lambda_i=0 \\
	&& u_{c_{i2}}-p\lambda_i=0 \label{2typeskey}\\
	&& \dot\mu_i =(\phi+\rho)\mu_i-u_{h_i} \\
	&& \dot \lambda_i =(\rho-r)\lambda_i \\
	&& \lim_{t\rightarrow\infty}h_i\mu_i e^{-\rho t}=0 \\
	&& \lim_{t\rightarrow\infty} b_i \lambda_i e^{-\rho t}=0
\end{eqnarray}

To find the decentralized equilibrium, we need now to aggregate. Observe that we can easily find the aggregate version of equation (\ref{2typeskey}). Since all the workers are inherently identical and the first order conditions are the same across agents we have that:
$$\int_0^\xi u_{c_{12}} -p\lambda_1 dn + \int_\xi^1 u_{c_{22}} -p\lambda_2 dn=0$$
which taking into account that both types of agents have the same linear-quadratic utility function
\begin{small}
$$
\int_0^\xi a_{c_2}+a_{c_2c_2}c_{12}+a_{c_1c_2}c_{11}+a_{c_2h}h_1 dn+ \int_\xi^1 a_{c_2}+a_{c_2c_2}c_{22}+a_{c_1c_2}c_{21}+a_{c_2h}h_2 dn   -p(\lambda_1+\lambda_2)=0
$$
\end{small}
and therefore
$$
a_{c_2}+a_{c_2c_2}\underbrace{(\xi c_{12}+(1-\xi)c_{22})}_{\equiv c_2}+a_{c_1c_2}\underbrace{(\xi c_{11}+(1-\xi)c_{21})}_{\equiv c_1}+a_{c_2h}\underbrace{(\xi h_1+(1-\xi) h_2)}_{\equiv h}-p\underbrace{(\lambda_1+\lambda_2)}_{\equiv\lambda}=0
$$
where $c_1$, $c_2$, $h$ and $\lambda$ indicate the aggregate variables of the corresponding variables. Solving now for $p$ we get
\begin{equation}
	p=\frac{a_{c_2}+a_{c_2c_2}c_2+a_{c_1c_2}c_1+a_{c_2 h}h}\lambda=\frac{\phi(a_{c_2}+a_{c_2c_2}y_2)+[(\psi_1+\phi)a_{c_1c_2}+\phi a_{c_2 h}]h}{\phi\lambda} \label{price_equation1}
\end{equation}
where the last espression has been obtained imposing the goods market clearing conditions $c_2=y_2$ and
$\dot b+c_1=rb+ y_1.$
Therefore, the same price equation as in the model with one agent holds.

\section{Empirical evidence}\label{CH:Empirical}

Our theory predicts that habits formed on good 1 can alter the relative price of the other good, see for example equations (\ref{eqeq:p}) and (\ref{price_equation1}). To the best of our knowledge, the existing empirical contributions on habits formation have not yet analysed this relation as the focus has always been, since the seminal works of Deaton \cite{Deaton} and Muellbauer \cite{muel}, on the effect of habits formed on one good on the consumption of the same good and consequently on its price. In other words, a common assumption in the literature is that the utility function is separable in the different goods and habits.\footnote{Ravn et al. \cite{morten} in a context of deep habits show that the demand function is affected by the relative prices however they do not estimate the demand function as we do next.} Therefore, the purpose of this section is to provide supporting empirical evidence to our theory.

To this purpose, we consider the CPI of admission to movies, theatres, and concerts and the CPI of cable and satellite television service in U.S. city average for all urban consumers in the period 2003-2019 with index 2003=100. These two time-series are available from the U.S. Bureau of Labor Statistics. The ratio of these two variables is used as a proxy of the relative price for these two different activities. 
$$RPI_t \ = \ \frac{CPI_t (Movies, \ theatres, \ concerts)}{CPI_t (cable \  and \ satellite \ TV \ services)}$$
where $RPI_t$ stands for relative price index.
               
\begin{center}
	\begin{figure}[!h]
		\centering 
		\begin{subfigure}[b]{0.7\textwidth}
			\includegraphics[width=\textwidth]{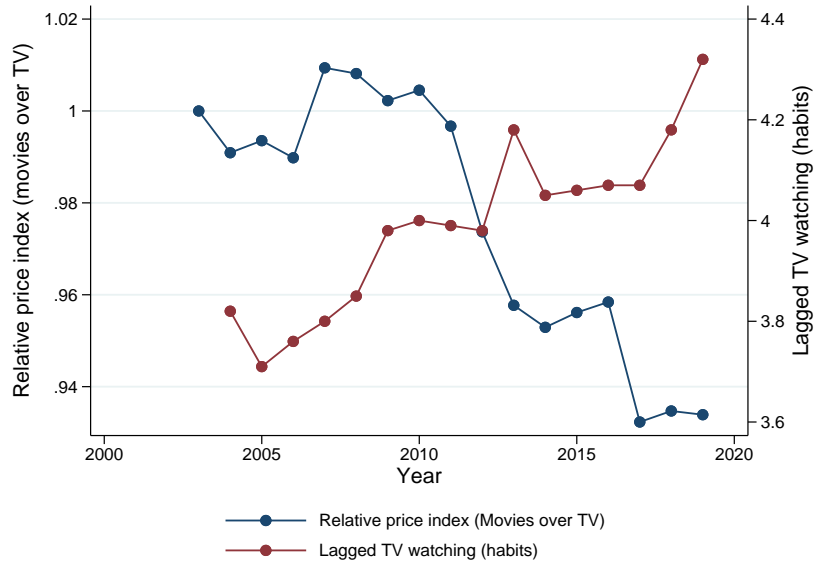}\caption{}
		\end{subfigure}
		\hfill	
		\begin{subfigure}[b]{0.7\textwidth}
			\includegraphics[width=\textwidth]{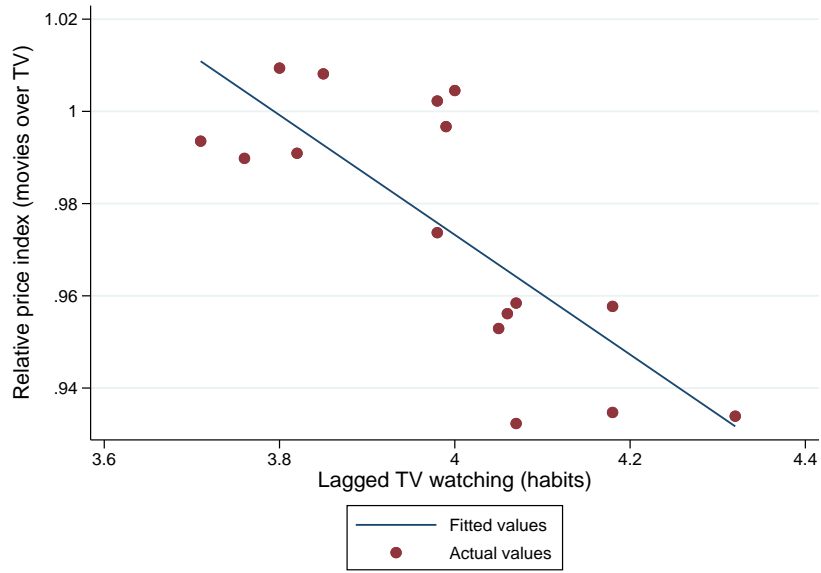}\caption{   }
		\end{subfigure}
		\caption{Relative price index and lagged average TV watching  in hours per day (habits) during weekends  (a) Time series. (b) scatter plot and correlation}\label{figEVIDENCE}
	\end{figure}
\end{center}

\vspace{-1cm}

Then we have used the American Time Use Survey (ATUS) to find the average hours per day in year $t$ of TV watching across persons who engaged in the activity (from now on $TV \ Watching \ _t $). This value is available for every year in the period $2003-2019$. We have also focused on the time spent on this activity during the weekend and holidays because the distribution of the time spent in consuming the other goods (movies etc) suggests that more than  40\% of it is concentrated in the weekend.\footnote{See, for example, the distribution of cinema attendance in the UK in 2013, by weekday reported by the UK Cinema Association. The data is available at the following webpage https://www.statista.com/statistics/296245/cinema-attendance-distribution-uk-by-weekday/. } Following Deaton \cite{Deaton}, the habits are defined as the lagged value of the consumption of an activity. Therefore in our case, $TV \ Watching \ _{t-1}$ captures the habits of watching TV.\footnote{In particular, Deaton and many contributions  on the empirical literature consider the case $h_{t+1}=c_t$ or equivalently $\Delta h_{t+1}=c_t-\delta_t h_t$ with $\delta=1$. This is a quite extreme case as the habits fully depreciate every period.}   

Figure \ref{figEVIDENCE} a) shows the $RPI_t$ and $TV \ Watching \ _{t-1}$ time series while Figure \ref{figEVIDENCE} b) is a scatter plot of the two series. As it emerges quite clearly there is a negative correlation which does not seem driven by any particular outlier.   

Then we have run some regressions to confirm our preliminary result and also to control for other variables: 
$$RPI_t=\beta_0+\beta_1 \cdot Watching \ TV \ _{t-1} + \beta_2 X_t + \varepsilon_t$$
In particular, we have controlled for the following other two variables: average hours per day in each year 2013-2019 spent for arts and entertainments (other than sports) during weekend and holidays across persons engaged in the activity, and average hours per day in each year 2016-2019  spent for watching movies in theatres during weekend and holidays across persons engaged in the activity. The last two variables are also lagged to capture the habits in such activities.   

The results of all the regressions are reported in table 1. The variable Watching TV (lagged) capturing the habits formed in doing this activity is always significant. Interestingly the second regression shows that Watching TV, which is a proxy of  consumption, is not significant while the habits are although less than in the other regressions. Arts and entertainments (lagged) and Movies watching at the theatre (lagged) are both not significant suggesting that the habits formed in Watching TV are the main driver of the relative price index. 
Similar results can be obtained if we consider i) the same time series but on working days instead of weekends; ii) the same time series but looking at the whole population and not just the persons engaged in the activity; iii) the same controls but in level and not in their lagged value.

\begin{table}[htbp]\centering
	\caption{Dependent Variable: Relative Price Index}
	\begin{tabular}{c c c c c c c c}
		\toprule
& (1)	&  & (2) &  & (3) &  & (4) \\ 
		\midrule
		
		Watching TV (lagged)  &-0.13***&  &        -0.091*       &  &       -0.127***   &  &       -0.16***  \\ 
		& (0.02)&&       (0.05)        &  &      (0.032)   &  &      (0.04)   \\ \\
		Watching TV        &         &  & -0.046      &  &     \\
		&         &  &   (0.05)      &  &         \\ \\
		Arts and entertainments (lagged)   &&  &               &  &             -0.005  &  &      \\
		&    &&           &  &        (0.03)       &  &         \\ \\
		Movies (lagged)  &&   &               &  &               &  &       -0.04 \\
		&   &&            &  &               &  &      (0.06)   \\ \\
		Constant      &1.49***& &       1.52***&  &       1.49***&  &       1.74***\\
		& (0.11)&&     (0.03)   &  &      (0.12)   &  &      (0.16)   \\ \\
		\midrule
		No. of observations    &16&       &     16   &  &     16   &  &     13   \\          
		Adj R-squared          &0.57&  &       0.56   &  &       0.54   &  &       0.63   \\
		\bottomrule
		\addlinespace[1ex]
		\multicolumn{8}{l}{\small{\textit{Notes:} Standard errors are reported in parenthesis. \textsuperscript{***}, 			\textsuperscript{**}, and \textsuperscript{*} indicate significance }} \\ \multicolumn{8}{l}{\small{at the 1\%, 5\%, and 10\% respectively.}}
	\end{tabular}
\end{table}

Overall, these findings are consistent with and support the theoretical relation between relative price and habits predicted by our model. This example is also interesting because the insignificance of the habits on art and entertainment, and movies is coherent with our assumption of considering habits formation on just one good.

\section{Conclusion} \label{Sec:Conclusion}

In this paper, we have shown how habits and goods' demand can change not just during but also after a lockdown. A substitutability and a satiation effect as well as the length of the lockdown are key to understand the direction and the intensity of this change. Given a sufficiently long lockdown, we find that a sector of the economy could shut down forever if the satiation effect dominates the substitutability effect. Government policies oriented to prevent this outcome have been discussed. On the other hand, a pent-up demand can be formed during the lockdown which may lead to a strong demand and higher prices once the lockdown is over if the substitutability effect dominates the satiation effect. Finally extensions of our model have been considered to investigate the effects of a permanent change in labor composition, the goodness of announcing a lockdown duration, as well as the consequences of keeping it uncertain. An extension with two types of agents and another with a minimum provision of the good produced by the firms affected by the lockdown lead to similar results.

\newpage
\begin{appendices}

\section*{APPENDIX (For Online Pubblication)}

\section*{Appendix A: Proofs and other theoretical results}
\begin{proof}[Proof of Proposition \ref{Prop:SS}]

Given the costate variable $\lambda$, we define
\[S_\lambda=\begin{cases}
	S^{ben}_\lambda\,\,\,\,\,\textrm{if habits are beneficial},\\
	S^{harm}_\lambda\,\,\textrm{if habits are harmful}.\\
\end{cases}\]
where $S^{ben}_\lambda$, $S^{harm}_\lambda$ are respectively the convex envelope of the points
\[\begin{bmatrix}
0\\
y\\
\frac{\phi+\rho}{\phi}\lambda_0
\end{bmatrix},\,\,
\begin{bmatrix}
0\\
y\\
0
\end{bmatrix},\,\,
\begin{bmatrix}
\lambda_0\\
y\\
0
\end{bmatrix}\]
with $y\geq 0$, and
\[\begin{bmatrix}
x\\
0\\
(\lambda_0-x)\frac{\phi+\rho}{\phi}
\end{bmatrix},\,\,
\begin{bmatrix}
x\\
0\\
0
\end{bmatrix},\,\,
\begin{bmatrix}
\lambda_0\\
0\\
0
\end{bmatrix}\]
with $x\geq \lambda_0$.

We will now prove that if there exists some $h>0$ such that $\nabla u (h,h,y_2)\in S_\lambda$ and some other $\tilde{h}$ such that
\begin{equation}\label{prop1_sufficient_condition}
\nabla u (\tilde{h},\tilde{h},y_2)\notin S_\lambda ,\end{equation}
 then a steady state exists with all the stationary variables function of the costate variable $\lambda$.
	
We look for a constant solution of the system \eqref{b_state}-\eqref{tvc:2}, when the market clearing conditions are also considered. The variables $(c_1,c_2,h,b,p,\lambda,\mu)$ satisfy the following conditions
\begin{eqnarray}
	&& c_1=rb+y_1 \label{b_constant}\\
	&& c_2=y_2 \label{c_2_constant}\\
	&& c_1=h \label{c_1_constant}\\
	&& u_{c_1}=-\mu\phi+\lambda\label{u_grad1}\\
	&& u_{c_2}=p\lambda \label{u_grad2}\\
	&& u_h=(\phi+\rho)\mu \label{u_grad3}\\
	&& \lambda=\lambda_0.\label{lambda_constant}
\end{eqnarray}
where $y_1,y_2$ are constants and $c_1$ and $c_2$ are related to them through the goods market clearing conditions  \eqref{lmrktclear},\eqref{cmrktclear}. In particular, we have that in equilibrium
\[c_2=y_2, \qquad and \qquad  b=\frac{h-y_1}{r}.\]
Then, the steady state value of $c_1$ and $b$ can be found by determining the steady state value of $h$.
In order to solve the second part of the system, \eqref{u_grad1}-\eqref{u_grad3} we will follow a geometric approach. We start by observing that the function $(c_1,c_2,h)\to \nabla u(c_1,c_2,h)$ is a curve in $\mathbb{R}^3$.
Since $c_1=h$, we can think the utility function as a function of $h$, where $c_2=y_2$ is a parameter. Thus, we call $\Psi_{y_2}$, the curve defined as
\begin{eqnarray*}
&\Psi_{y_2}:&\mathbb{R}^+\longmapsto \mathbb{R}^3\\
&&h\longmapsto \Psi_{y_2}(h):=\nabla u(h,y_2,h).\\
\end{eqnarray*}
Moreover, the variable $\lambda\equiv \lambda_0$ is constant as given in \eqref{lambda_constant}.
When the variables $\mu,p$ vary in $\mathbb{R}^2$,
the right hand side of \eqref{u_grad1}-\eqref{u_grad3} identifies a plane whose parametric equation is:
\begin{eqnarray*}
&\Phi_{\lambda_0}:&\mathbb{R}^2\longmapsto \mathbb{R}^3\\
&&(\mu, p)\longmapsto \Phi_{\lambda_0}(\mu,p):=\begin{pmatrix}
-\mu\phi+\lambda_0\\
p\lambda_0\\
(\phi+\rho)\mu
\end{pmatrix}.\\
\end{eqnarray*}
Since $u_{c_i}>0$, then we can reformulate \eqref{u_grad1}-\eqref{u_grad3} as
\[\Psi_{y_2}(h)=\Phi_{\lambda_0}(\mu,p).\]
If the curve $\Psi_{y_2}$ and the plane $\Phi_{\lambda_0}$ intersect in a unique point, then a unique triplet $(h,\mu, p)$ solving  \eqref{u_grad1}-\eqref{u_grad3} exists. Consequently, the system \eqref{b_constant}-\eqref{lambda_constant} is uniquely solved.
If the habits are beneficial the curve stays in first orthant of $\mathbb{R}^3$ , otherwise, the curve belongs to the fifth orthant.
The set $S_{\lambda_0}$ is the area of the orthant delimited by the plane $\Phi_{\lambda_0}(\mu,p)$: in the beneficial case, it is the area of the first orthant , below the plane $\Phi_{\lambda_0}(\mu,p)$ (see figure \ref{fig:plane}); in the harmful case, it is the area of the fifth orthant above the plane $\Phi_{\lambda_0}(\mu,p)$ (see figure \ref{fig:plane2}).
Since the curve is continuous and there exists $h$ and $\tilde{h}$ such that $\nabla u(h,h,y_2)\in S_\lambda$, and $\nabla u(\tilde{h},\tilde{h},y_2)\notin S_\lambda$, then the curve must intercept the plane in, at least, one point.
From the geometrical interpretation, the proof of uniqueness is straightforward. Indeed, if for $h>h^*$ the curve runs below (resp. above) the plane and for $h>h^*$ the curve runs above (resp. below) the plane, then it is obvious to observe that the curve intercepts the plane in a unique point, and therefore it exists a unique stationary solution.
\begin{figure}
\centering
\begin{tikzpicture}
\draw[thick,-stealth] (0,0,0)--(0,0,6);
\draw[thick,-stealth] (0,0,0)--(0,6,0);
\draw[thick,-stealth] (0,0,0)--(6,0,0);
\coordinate[] (O) at (0,0,0);
\coordinate[] (A) at (0,0,4);
\coordinate[] (B) at (6,0,4);
\coordinate[] (C) at (6,3.5,0);
\coordinate[] (D) at (0,3.5,0);
\coordinate[label=$P$] (P) at (2,3.5-1.5*3.5/4,1.5);
\coordinate[label=$Q$] (Q) at (3,0,0);
\filldraw[black] (P) circle (2pt);
\coordinate[label=$x$] (x) at (0,0,6);
\coordinate[label=$y$] (x) at (6,0,0);
\coordinate[label=$z$] (z) at (0,6,0);
\coordinate[label=$\Psi_{\lambda_0}$] (psi) at (6.5,3.5,0);
\coordinate[label=$\Phi_{y_2}$] (psi) at (3,1,0);
\filldraw[black] (A) circle (2pt) node[anchor=east] {$(\lambda_0,0,0)$};
\filldraw[black] (D) circle (2pt) node[anchor=east] {$(0,0,\frac{\phi+\rho}{\phi}\lambda_0)$};
\draw[color=red, dashed] (Q) to [bend left=-30] (P);
\draw[color=red] (P) to [bend left=-30] (2,2.5,5);
\draw[fill=brown,opacity=0.3] (A)--(B)--(C)--(D);
\draw[thick] (A) -- (D);
\draw[thick] (A) -- (B);
\draw[thick] (C) -- (D);
\draw[dashed] (B) -- (8,0,4);
\draw[dashed] (C) -- (8,3.5,0);
\end{tikzpicture}
\caption{Illustration of the proof of Proposition 1 in the beneficial case}\label{fig:plane}
    \end{figure}
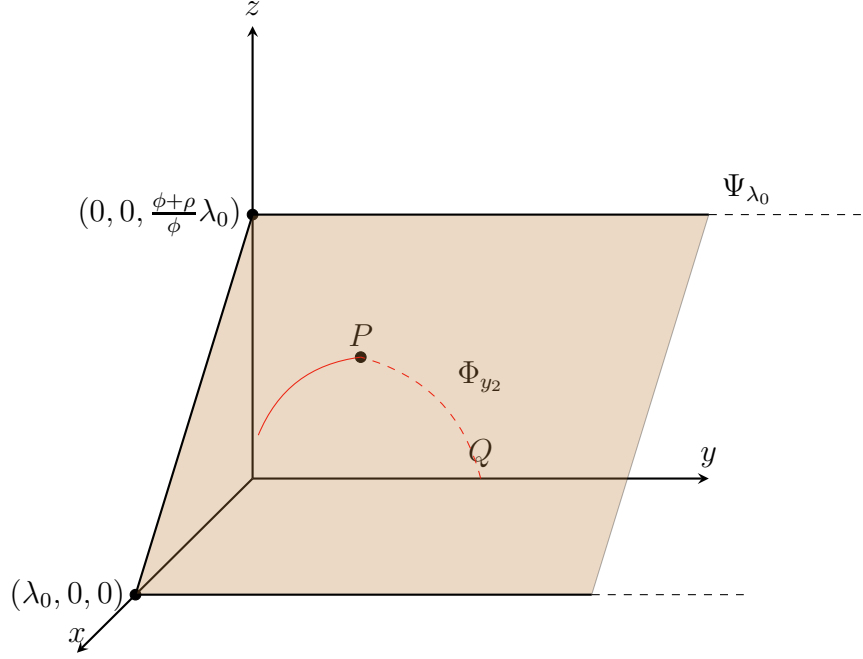

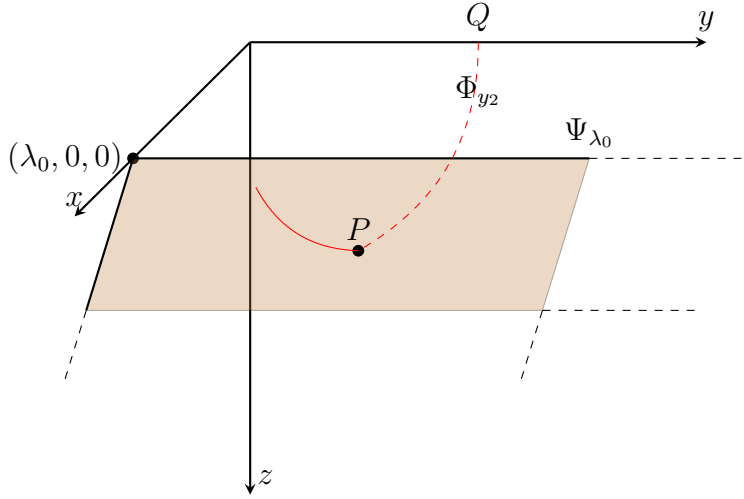
\begin{figure}
\centering
\begin{tikzpicture}
\draw[thick,-stealth] (0,0,0)--(0,0,6);
\draw[thick,-stealth] (0,0,0)--(0,-6,0);
\draw[thick,-stealth] (0,0,0)--(6,0,0);
\coordinate[] (O) at (0,0,0);
\coordinate[] (A) at (0,0,4);
\coordinate[] (E) at (0,-0.4*3.5,1.4*4);
\coordinate[] (B) at (6,0,4);
\coordinate[] (F) at (6,-0.4*3.5,1.4*4);
\coordinate[label=$x$] (x) at (0,0,6);
\coordinate[label=$y$] (x) at (6,0,0);
\coordinate[label=$z$] (z) at (0.2,-6,0);
\filldraw[black] (A) circle (2pt) node[anchor=east] {$(\lambda_0,0,0)$};
\draw[fill=brown,opacity=0.3] (A)--(B)--(F)--(E);
\draw[thick] (A) -- (E);
\draw[thick] (A) -- (B);
\draw[dashed] (B) -- (8,0,4);
\draw[dashed] (F) -- (8,-0.4*3.5,1.4*4);
\draw[dashed] (F) -- (6,-0.6*3.5,1.6*4);
\draw[dashed] (E) -- (0,-0.6*3.5,1.6*4);
\coordinate[label=$P$] (P) at (2,-3.5+1.5*3.5/4,1.5);
\coordinate[label=$Q$] (Q) at (3,0,0);
\filldraw[black] (P) circle (2pt);
\coordinate[label=$\Psi_{\lambda_0}$] (psi) at (6,0,4);
\coordinate[label=$\Phi_{y_2}$] (psi) at (3,-1,0);
\draw[color=red, dashed] (Q) to [bend left=30] (P);
\draw[color=red] (P) to [bend left=30] (2,0,5);
\end{tikzpicture}
\caption{Illustration of the proof of Proposition 1 in the harmful case}\label{fig:plane2}
    \end{figure}
\end{proof}
\begin{remark}
The sufficient condition presented in Proposition \ref{Prop:SS}, namely condition \eqref{prop1_sufficient_condition}, is trivially verified when one of the asymptotic behaviour listed below occurs. If the habits are beneficial, then the sufficient condition is verified if one the following conditions holds:
\begin{itemize}
    \item $\lim_{c_1\to +\infty}u_{c_1}(x)=\lim_{h\to +\infty}u_h(x)=0$ and $\lim_{h\to 0} u_h(x)=+\infty$;
    \item $\lim_{c_1\to +\infty}u_{c_1}(x)=\lim_{h\to +\infty}u_h(x)=0$ and $\lim_{c_1\to 0} u_{c_1}(x)=+\infty$;
    \item $\lim_{c_1\to 0}u_{c_1}(x)=\lim_{h\to 0}u_h(x)=0$ and $\lim_{h\to 0} u_h(x)=+\infty$;
    \item $\lim_{c_1\to 0}u_{c_1}(x)=\lim_{h\to 0}u_h(x)=0$ and $\lim_{c_1\to +\infty} u_{c_1}(x)=+\infty$.
\end{itemize}
Notice that the first two conditions are actually Inada conditions, while the other conditions refer to cases where the Inada conditions do not hold.
If the habits are harmful, then the sufficient condition is verified if one the following holds,
\begin{itemize}
      \item $\lim_{c_1\to 0}u_{c_1}(x)=+\infty,\lim_{h\to 0}u_h(x)=0$ and $\lim_{c_1\to +\infty} u_{c_1}(x)=0$;
\item $\lim_{c_1\to +\infty}u_{c_1}(x)=+\infty,\lim_{h\to +\infty}u_h(x)=0$ and $\lim_{c_1\to 0} u_{c_1}(x)=0$;
\end{itemize}
Again, the Inada conditons hold in the first case but not in the second.
\end{remark}

\medskip

\begin{proof}[Proof of Proposition \ref{Prop:LocalDynamics}]
	Linearization of (\ref{foc:c1}) around the steady state leads to $$u^*_{c_1c_1}\tilde c_1+u^*_{c_1 h} \tilde h +\phi\tilde\mu=0.$$
	Rearrenging the terms leads to (\ref{eqeq:c1}). Linearizing now (\ref{foc:h}) we have that
	$$\dot{\tilde\mu}=(\phi+\rho)\tilde\mu - u^*_{hc_1}\tilde c_1-u^*_{hh}\tilde h;$$
	substituting the value of $\tilde c_1$ previously found, i.e.  (\ref{eqeq:c1}), leads to (\ref{eqeq:mu1}). On the other hand, linearizing the habit equation and substituting (\ref{eqeq:c1}) leads to  (\ref{eqeq:h1}). (\ref{eqeq:b1}) has been obtained by linearizing the market clearing condition of good 1 around the steady state while equation (\ref{eqeq:p1}) by linearizing (\ref{foc:c2}) around the steady state and solving for $\tilde p$.
\end{proof}

\medskip

\begin{lemma}\label{Lemma:eigenvalues}
Let us write the system (\ref{eqeq:mu1}) and (\ref{eqeq:h1}) in matrix form:
\begin{equation}\left(\begin{array}{c}
		\dot {\tilde\mu} \\ \\ \dot {\tilde h} \end{array}\right)=\underbrace{\left(\begin{array}{cc} \left[\left(1+\frac{u^*_{c_1h}}{u^*_{c_1c_1}}\right)\phi+\rho\right] &  \frac{(u^*_{c_1h})^2-u^*_{c_1c_1}u^*_{hh}}{u^*_{c_1c_1}} \\
			\\	-\frac{\phi^2}{u^*_{c_1c_1}} & -\phi\left(1+\frac{u^*_{c_1h}}{u^*_{c_1c_1}}\right) \end{array}\right)}_{=\mathbf A} \left(\begin{array}{c}
		\tilde\mu \\ \\ \tilde h \end{array}\right) \label{system:hmu}
\end{equation}
The eigenvalues of matrix $\mathbf A$ are
	\begin{equation}\label{eq:eigenvalues}
		\psi_i=\frac{\rho\pm\sqrt{\rho^2+\frac{4\phi}{u^*_{c_1c_1}}[(\phi+\rho)u^*_{c_1c_1}+(\rho+2\phi)u^*_{c_1h}+\phi u^*_{hh}]}}2.	
	\end{equation}
	
	and they have the following properties:
	\begin{itemize}
		\item they are real and have opposite sign if $u^*_{c_1h}<-\frac{(\phi+\rho)u^*_{c_1c_1}+\phi u^*_{hh}}{\rho+2\phi}\equiv \bar u^*_{c_1h} $;
		\item they are real and have positive sign if $-\frac{(\phi+\rho)u^*_{c_1c_1}+\phi u^*_{hh}}{\rho+2\phi}<u^*_{c_1h}<-\frac{\rho^2u^*_{c_1c_1}}{4\phi(\rho+2\phi)}-\frac{(\phi+\rho)u^*_{c_1c_1}+\phi u^*_{hh}}{\rho+2\phi}$;
		\item they are conjugate-complex with positive real part if $u^*_{c_1h}>-\frac{\rho^2u^*_{c_1c_1}}{4\phi(\rho+2\phi)}-\frac{(\phi+\rho)u^*_{c_1c_1}+\phi u^*_{hh}}{\rho+2\phi}$.
	\end{itemize}
\end{lemma}
\begin{proof}[Proof of Lemma \ref{Lemma:eigenvalues}]
	The determinant and trace of $\mathbf A$ are $$Det (\mathbf A)=-\frac{\phi}{u^*_{c_1c_1}}[(\phi+\rho)u^*_{c_1c_1}+(\rho+2\phi)u^*_{c_1h}+\phi u^*_{hh}] \qquad and \qquad Tr(\mathbf A)=\rho $$ and therefore the discriminant is $$\Delta(\mathbf A)=\rho^2+\frac{4\phi}{u^*_{c_1c_1}}[(\phi+\rho)u^*_{c_1c_1}+(\rho+2\phi)u^*_{c_1h}+\phi u^*_{hh}]$$
Observe that the sign of the determinant depends on $u_{ch}$ to be lower or higher than $$-\frac{(\phi+\rho)u^*_{c_1c_1}+\phi u^*_{hh}}{\rho+2\phi},$$ while the discriminant sign depends on $u_{ch}$ to be lower or higher than $$-\frac{\rho^2u*_{c_1c_1}}{4\phi(\rho+2\phi)}-\frac{(\phi+\rho)u^*_{c_1c_1}+\phi u^*_{hh}}{\rho+2\phi}.$$

Finally, the result follows immediately taking into account that $Det(\mathbf A)=\psi_1\psi_2$, $Tr(\mathbf A)=\psi_1+\psi_2$ and that $\Delta(\mathbf A)>0$.
\end{proof}

\medskip

\begin{proof}[Proof of Proposition \ref{Prop:solsystem}]
	First the general solution of the homogenous linear ODEs system is
	\begin{eqnarray*}
		\tilde \mu=d_1 v_{11} e^{\psi_1 t}+ d_2 v_{21} e^{\psi_2 t} \\
		\tilde h=d_1 v_{12} e^{\psi_1 t}+ d_2 v_{22} e^{\psi_2 t}
	\end{eqnarray*}
	where $\mathbf v_i =[v_{i1} \ v_{i2}]$ is the eigenvector associated to the eigenvalue $\psi_i$. Assume without loss of generality that $\psi_1$ is the negative eigenvalue and $\psi_2$ is the positive eigenvalue. It is immediate to see that the TVC (\ref{tvc:1}) holds only if $d_2=0$. Therefore we have that
	\begin{eqnarray*}
		\tilde h=\tilde h_0 e^{\psi_1 t} \\
		\tilde \mu= \frac{v_{11}}{v_{12}} \tilde h
	\end{eqnarray*}
	We need now to find the eigenvalue $\mathbf v_1 =[v_{11} \ v_{12}]$. As usual we have that
	$$(\mathbf A-\psi_1 \mathbf I) \mathbf v'_1=\mathbf 0$$
	from which we find that
	$$-\frac{\phi^2}{u^*_{c_1c_1}}v_{11}-\left[\phi\left(\frac{u^*_{c_1h}}{u^*_{c_1c_1}}+1\right)+\psi_1\right]v_{12}=0$$
	and therefore
	$$\frac{v_{11}}{v_{12}}=-\frac{\phi u^*_{c_1h}+(\phi+\psi_1)u^*_{c_1c_1}}{\phi^2}$$
\end{proof}

\medskip

\begin{proof}[Proof of Proposition \ref{Prop:solsystem_b}]
	First substitute (\ref{eqeq:c1h}) into (\ref{eqeq:b1}) and solve the initial value problem
	$$\dot{\tilde b}=r\tilde b- \frac{\phi+\psi_1}\phi \tilde h_0 e^{\psi_1 t}$$
	with $\tilde b(0)=\tilde b_0$ given. Its solution is
	$$\tilde b=\left(\tilde b_0-\frac{\phi+\psi_1}{\phi(r-\psi_1)}\tilde h_0\right)e^{r t}+\frac{\phi+\psi_1}{\phi(r-\psi_1)} \tilde h_0 e^{\psi_1 t}$$
	Moreover, the TVC (\ref{tvc:2}) is respected as long as $\lambda$ is chosen so that the term in parenthesis is equal to zero. In fact, TVC (\ref{tvc:2}) rewrites
	$$\lim_{t\rightarrow\infty} b \lambda e^{-r t}=\lambda b^* e^{-\rho t} +  \lambda\left(\tilde b_0-\frac{\phi+\psi_1}{\phi(r-\psi_1)}\tilde h_0\right)+\lambda\frac{\phi+\psi_1}{\phi(r-\psi_1)} \tilde h_0 e^{(\psi_1-r) t}                  =0$$
	which holds if and only if
	$$\tilde b_0=\frac{\phi+\psi_1}{\phi(r-\psi_1)}\tilde h_0$$. 	
\end{proof}

\medskip

\begin{proof}[Proof of Proposition \ref{Prop:prelock}]
	At the steady state we have that $\dot h=\dot b=\dot\mu=0$. Therefore, from the habits formation equation we immediately see that $c^*_1=h^*$. Using now equation (\ref{foc:c1}) and (\ref{foc:c2}) we find respectively that
	\begin{eqnarray}\label{mu_*}
	&& \mu^*=\frac{\lambda-a_{c_1}-(a_{c_1c_1}+a_{c_1h})h^*-a_{c_1c_2}y_2}\phi, \\
&& p^*=\frac{a_{c_2}+a_{c_2c_2}y_2+(a_{c_1c_2}+a_{c_2h})h^*}\lambda\end{eqnarray}
	Substituting the first equation into equation (\ref{foc:h}) we get
	\begin{equation}h^*= \frac{(\phi+\rho)\left[\lambda-\left(a_{c_1c_2}+\frac{\phi}{\phi+\rho}a_{c_2h}\right)y_2-\frac\phi{\phi+\rho} a_h-a_{c_1}\right]}{(\phi+\rho)a_{c_1c_1}+(2\phi+\rho)a_{c_1h}+\phi a_{hh}}.\end{equation}
	In addition, the good 1 market clearing condition (\ref{cmrktclear}) implies that
\begin{equation} rb^*=h^*-y_1. \label{bstarhstar} \end{equation}
	By Lemma \ref{Lemma:eigenvalues} and Proposition \ref{Prop:solsystem} we know that if $a_{c_1h}<-\frac{(\phi+\rho)a_{c_1c_1}+\phi a_{hh}}{\rho+2\phi}$ then the eigenvalues are real and have opposite sign.\footnote{In the case of a linear-quadratic utility function we do not need to linearize the system to get (\ref{eqeq:mu1})-(\ref{eqeq:p1}) since the FOCs lead to a system of linear no-homogeneous ODEs.} In addition, we can find the solution of the habits stock and the costate variable $\mu$ as deviation from their steady state value. Using then Proposition \ref{Prop:solsystem_b} we can find the value of $\lambda$ by solving the equation
	\begin{equation} b_0-b^*=\frac{\phi+\psi_1}{\phi(r-\psi_1)}(h_0-h^*) \label{b0hstar}\end{equation}
	substituting the value of $h^*$ and $b^*$ found previously we find that
	\begin{equation}
		\lambda=m_0+m_1\left[rb_0+y_1+\frac{r(\phi+\psi_1)}{\phi(\psi_1-r)}h_0\right] \label{lambdainproof}
	\end{equation}
	where
	\begin{eqnarray*}
		&& m_0= \frac{(\phi+\rho)a_{c_1}+\phi a_h+[(\phi+\rho)a_{c_1c_2}+\phi a_{c_2h}]y_2}{\phi+\rho}\\
		&& m_1= \frac{\phi(\psi_1-r)}{(\phi+r)\psi_1}\cdot\frac{(\phi+\rho)a_{c_1c_1}+(\rho+2\phi)a_{c_1h}+\phi a_{hh}}{\phi+\rho}<0
	\end{eqnarray*}
	Substituting back the value of $\lambda$ into the steady state equation of $h^*$, $p^*$, $b^*$ leads to our result. Notice that the steady state equation of $p^*$ can be found using equation (\ref{foc:c2}).
	
	Before proceeding let us define the following thresholds which will turn out usefull in the rest of the proof:
		\begin{small}
		$$\bar a_{c_1h}\equiv -\frac{(\phi+\rho)a_{c_1c_1}+\phi a_{hh}}{\rho+2\phi}, \qquad \underline a_{c_1}\equiv -\frac{\phi a_h+[(\phi+\rho)a_{c_1c_2}+\phi a_{c_2h}]y_2}{\phi+\rho},$$ $$ \underline b_0\equiv -\frac{y_1}r+\frac{\phi+\psi_1}{\phi(r-\psi_1)}h_0, \qquad \bar b_0\equiv \underline b_0-\frac{m_0}{rm_1}, \qquad \underline h_0\equiv\frac{\phi(r-\psi_1)}{r(\phi+\psi_1)}\left[y_1+\frac{m_0}{m_1}\right],\ \ and $$ $$ \underline a_{c_2}\equiv -a_{c_2c_2}y_2+\frac{\tau m_0}{(1-\alpha)y_2}+\left[\frac{\tau m_1}{(1-\alpha)y_2}-\frac{\phi(\psi_1-r)(a_{c_1c_2}+a_{c_2h})}{(\phi+r)\psi_1}\right]\left[rb_0+y_1+\frac{r(\phi+\psi_1)}{\phi(\psi_1-r)}h_0\right].$$
	\end{small}
	
	We can now proceed and find the condition which guarantees that $h^*>0$, $\lambda>0$ and $p^*\geq \frac\tau{(1-\alpha)y_2}$, the latter guarantees that also sector 2 is active.
	
	Let us start with $h^*>0$. Combining (\ref{bstarhstar}) with (\ref{b0hstar}) we get that
	$$h^*=\underbrace{\frac{r\phi(r-\psi_1)}{\psi_1(r+\phi)}}_{<0}\left[-b_0-\frac{y_1}r+\frac{\phi+\psi_1}{\phi(r-\psi_1)}h_0\right]$$
	It is then immediate to see that $h^*>0$ as long as $b_0>\underline b_0$. Notice also that $\underline b_0>0$ as long as $h_0>\underline h_0$. On the other hand, we can see from (\ref{lambdainproof}) that
	$$\lambda>0 \qquad \Leftrightarrow \qquad rb_0+y_1+\frac{r(\phi+\psi_1)}{\phi(\psi_1-r)}h_0<-\frac{m_0}{m_1}$$ since $m_0>0$ and $m_1<0$ when $a_{c_1}>\underline a_{c_1}$. Solving for $b_0$ leads to the condition $b_0<\underline b_0-\frac{m_0}{r m_1}$. Finally using the expression for $p^*$, substituting it into $p^*\geq \frac\tau{(1-\alpha)y_2}$ and solving for $a_{c_2}$ leads to the condition $a_{c_2}\geq \underline a_{c_2}$.
\end{proof}

\medskip

\begin{proof}[Proof of Proposition \ref{Prop:keyresult}]
	First, $h_0<h^*_{NL}$ guarantees that in the economy without the lockdown the price $p$ converges from above (below) to $p^*$ if $(\phi+\psi_1)a_{c_2c_1}+\phi a_{c_2 h}<0$ $(>0)$. We need now to show under which condition $p_{ t,AL}<p^*$. Notice that $$p_{ t,AL}-p^*=\frac{(\phi+\psi_1)a_{c_2c_1}+\phi a_{c_2 h}}{\phi\lambda}\left[\underbrace{h^*_L-h^*_{AL}}_{>0}+(\underbrace{h_0-h^*_L}_{<0})e^{\psi_1 \tilde t}\right]e^{\psi_1(t-\tilde t)}$$
	Assume now that $(\phi+\psi_1)a_{c_2c_1}+\phi a_{c_2 h}<0$ then
	$$p_{t,AL}<p^* \ \ \Leftrightarrow \ \ h^*_L-h^*_{AL}+(h_0-h^*_L)e^{\psi_1 \tilde t}>0 \ \ \Leftrightarrow \ \ \tilde t>\frac{\ln(h^*_L-h_0)-\ln(h^*_L-h^*_{AL})}{|\psi_1|}$$
	
	On the other hand, if  $(\phi+\psi_1)a_{c_2c_1}+\phi a_{c_2 h}>0$ then $p_{ t,AL}>p_{ t,NL}$ since	
	$$p_{ t,AL}-p_{ t,NL}=\frac{(\phi+\psi_1)a_{c_2c_1}+\phi a_{c_2 h}}{\phi\lambda}\left[(\underbrace{h^*_L-h^*_{AL}}_{>0})(\underbrace{1-e^{\psi_1 \tilde t}}_{>0})\right]e^{\psi_1(t-\tilde t)}$$
	
	  The result can be extended to the case $h_0>h^*_{NL}$. In particular, if $h_0>h^*_L$, the sign of $p_{t,AL}-p^*$ depends just on $(\phi+\psi_1)a_{c_2c_1}+\phi a_{c_2 h}$. Now we focus on the case $h^*_{NL}<h_0<h^*_L$. Since it can be easily proven that $\tilde t<0$ then
	\[\left[\underbrace{h^*_L-h^*_{AL}}_{>0}+(\underbrace{h_0-h^*_L}_{<0})e^{\psi_1 \tilde t}\right]>0\]
	for each $\tilde t>0$.
	In summary, we conclude that
	
	\begin{itemize}

 \item if the satiation dominates the substitutability effect, $a_{c_2h}<\bar a_{c_2h}$, $h_0>h^*_{NL}$ then
    \begin{equation}
        p_{t,AL}<p_{t,NL}<p^*;
    \end{equation}

%

%
%

\item if instead the substitutability dominates the satiation effect, $a_{c_2h}>\bar a_{c_2h}$, $h_0>h^*_{NL}$ then
    \begin{equation}p^*<p_{t,NL}<p_{t,AL};\end{equation}

%
%
\end{itemize}

\end{proof}
\begin{proposition}\label{prop:huguali} 	
		$h^*_{AL}=h^*_{NL}$ and $\lambda_{AL}=\lambda_{NL}$.	\end{proposition}
\begin{proof}
	$$h^*_{AL}= \frac{\phi(\psi_1-r)}{(\phi+r)\psi_1}\left[rb_{\tilde t}+y_1+\frac{r(\phi+\psi_1)}{\phi(\psi_1-r)}h_{\tilde t}\right]$$
Substituting the value of $b_{\tilde t}$ and $h_{\tilde t}$ from (\ref{eqeq:inibAL}) and (\ref{eqeq:inihAL}) we get
$$h^*_{AL}= \frac{\phi(\psi_1-r)}{(\phi+r)\psi_1}\left[-\frac{\psi_1(\phi+r)}{\phi(r-\psi_1)}h^*_L-y_{1,L}+y_1\right]$$
Substituting now the value of $h^*_L$ from equation (\ref{eqeq:hstarLock}) leads to
\begin{eqnarray*} h^*_{AL}&=&  \frac{\phi(\psi_1-r)}{(\phi+r)\psi_1}\left[\frac{\psi_1(\phi+r)}{\phi(\psi_1-r)}\left(\frac{\phi(\psi_1-r)}{(\phi+r)\psi_1}\left[rb_0+y_{1,L}+\frac{r(\phi+\psi_1)}{\phi(\psi_1-r)}h_0\right]\right)-y_{1,L}+y_1\right] \\ &=& \frac{\phi(\psi_1-r)}{(\phi+r)\psi_1}\left[rb_0+y_1+\frac{r(\phi+\psi_1)}{\phi(\psi_1-r)}h_0\right]=h^*_{NL}.
\end{eqnarray*}
Now we focus on the parameter $\lambda$. By Proposition \ref{Prop:prelock}, we know that
\[\lambda_{NL}=m_0+m_1\left[rb_0+y_1+\frac{r(\phi+\psi_1)}{\phi(\psi_1-r)}h_0\right]. \]
To calculate $\lambda_{AL}$, we just need to replace $b_0,h_0$ with $b_{\tilde t},h_{\tilde t}$, then
\[\lambda_{AL}=m_0+m_1\left[rb_{\tilde t}+y_1+\frac{r(\phi+\psi_1)}{\phi(\psi_1-r)}h_{\tilde t}\right]. \]
Since the economy was in lockdown when $t\in[0,\tilde t]$, then $h_{\tilde t},b_{\tilde t}$ satisfies respectively identities \eqref{eqeq:inihAL}-\eqref{eqeq:inibAL}.
Then, from \eqref{eqeq:inihAL}-\eqref{eqeq:inibAL} and using the definition of $h^*_{L}$ in \eqref{eqeq:hstarLock} we get that
\[
 rb_{\tilde t}+y_{1,L}+r\frac{\phi+\psi_1}{\phi(\psi_1-r)}h_{\tilde t}= -\frac{\psi_1(\phi+r)}{\phi(r-\psi_1)}h^*_L
 =rb_0+y_{1,L}+\frac{r(\phi+\psi_1)}{\phi(\psi_1-r)}h_0.
\]
Thus, $\lambda_{NL}=\lambda_{AL}$.
\end{proof}

\medskip

\begin{proof}
By standard arguments, one can prove that $v(\tau,b(\tau),h(\tau))=e^{-\rho \tau}v(0,b(\tau),h(\tau))$. Since, the law of $\tau$ is known, we rewrite the value function in the following way,
\begin{align*}
v(0,b,h)&=\max_{c_1,c_2}\int_0^\infty \lambda e^{-\delta t}\int_0^{t}e^{-\rho s}u(c_1,h,c_2)dsdt+\\
&\hspace{5cm}+\int_0^\infty \delta e^{-\delta t} e^{-\rho t}v(0,b(t),h(t))dt\\
&=\max_{c_1,c_2}\int_0^\infty e^{-\rho s}u(c_1,h_1,c_2)\left(\int_s^\infty \delta e^{-\delta t} dt\right)ds+\\
&\hspace{5cm}+\int_0^\infty \delta e^{-(\rho+\delta) t} v(0,b(t),h(t))dt\\
&=\max_{c_1,c_2}\int_0^\infty e^{-(\rho +\delta)s}u(c_1,h_1,c_2)ds+\\
&\hspace{5cm}+\int_0^\infty \delta e^{-(\rho+\delta) t} v(0,b(t),h(t))dt.\\
&=\max_{c_1,c_2}\int_0^\infty e^{-(\rho +\delta)s}\left[u(c_1,h(s),c_2)ds+\delta v(0,b(s),h(s))\right]ds.\\
\end{align*}
\end{proof}

\section*{Appendix B: Derivations and Other Results}

\subsection*{Extension: Permanent change in labor composition}

\textbf{Derivation of equation (\ref{interpretation1})}

The equilibrium  relative price equation for the model with linear-quadratic utility is
$$p=\frac{\phi(a_{c_2}+a_{c_2c_2}y_2)+[(\psi_1+\phi)a_{c_1c_2}+\phi a_{c_2 h}]h}{\phi\lambda},$$ 	
which is now a function $p=f(y_2(\xi(t)),\lambda(\xi(t)),h(\xi(t),t))$. Differentiating it with respect to $t$ and setting $dt=1$ leads to
\begin{equation}
	dp=\underbrace{\left[\frac1\lambda \left(a_{c_2c_2}\frac{\partial y_2}{\partial\xi}-p\frac{\partial\lambda}{\partial \xi}\right)+ SSE \cdot \frac{\partial h}{\partial \xi}\right] d\xi}_{Labor \ Composition \ Effect \ (LCE)}+SSE\cdot dh \label{interpretation_prelim}
\end{equation}
where taking into account the shadow price equation (\ref{lambda_LQ}) and the habits path (\ref{habitpath}) and its steady state (\ref{h*}):
\begin{equation}
\footnotesize
	LCE= \frac1\lambda \left[a_{c_2c_2}\frac{\partial y_2}{\partial\xi}-p\left(a_{c_1c_2}+\frac{\phi}{\phi+\rho}a_{c_2h}\right)\frac{\partial y_2}{\partial\xi}-pm_1\frac{\partial y_1}{\partial \xi}\right] + SSE (1-e^{\psi_1(t-\tilde t)})\cdot \frac{\phi(\psi_1-r)}{(\phi+r)\psi_1} \frac{\partial y_1}{\partial \xi} \label{interpretation1_prelim}
\end{equation}
	
Therefore, at $t=\tilde t$ when Sector 2 reopens, the last term is zero and we get equation (\ref{interpretation1}).
Moreover, taking into account the expression of the relative price steady state (\ref{pstar}) and differentiating it with respect to $\xi$ leads to
$$\frac{dp^*}{d\xi}=\frac1\lambda \left(a_{c_2c_2}\frac{\partial y_2}{\partial\xi}-p\frac{\partial\lambda}{\partial \xi}\right)+SSE\cdot\frac{dh^*}{d\xi}$$
or equivalently
$$\frac{dp^*}{d\xi}=\frac1\lambda \left[\underbrace{a_{c_2c_2}\frac{\partial y_2}{\partial\xi}}_{>0}-p\left(a_{c_1c_2}+\frac{\phi}{\phi+\rho}a_{c_2h}\right)\frac{\partial y_2}{\partial\xi}\underbrace{-pm_1\frac{\partial y_1}{\partial \xi}}_{>0}\right]+SSE\cdot\underbrace{\frac{dh^*}{d\xi}}_{>0}.$$
Therefore, it follows immediately that
$$SSE>0 \qquad \Leftrightarrow \qquad a_{c_2h}>\bar a_{c_2h} \qquad \Rightarrow \qquad \frac{dp^*}{d\xi}>0$$
since $-\frac{\phi+\rho}\phi a_{c_1c_2}<-\frac{\phi+\psi_1}\phi a_{c_1c_2}$ when $a_{c_1c_2}>0$ (case of substitute goods). Observe that, a substitutability effect stronger than  a satitation effect is only a sufficient but not necessary condition for a positive change in the relative prices due to a readjustment in the labor composition. Observe also that $p$ may converge either from below or from above to its steady state level depending on the size of the labor composition change.

Exactly the opposite happens when
$$a_{c_2h}<<-\frac{\phi+\rho}\phi a_{c_1c_2} \qquad \Rightarrow \qquad \frac{dp^*}{d\xi}<0 $$
since now both terms in the right hand side of the expression of $\frac{dp^*}{d\xi}$ are negative.

\medskip

We now want to understand what happens to the economy when we are not in these two extreme cases. To do that, we begin assuming that the change in the labor composition during the lockdown with a fraction of work $a\in(0,1)$ allocated from sector 2 to sector 1 is permanent. Therefore, after lockdown we will have that
$$
\ell_{1,AL}=\xi \bar \ell+a(1-\xi)\bar \ell, \ \ and \ \   \ell_{2,AL}=(1-a)(1-\xi)\bar \ell.
$$
implying the following productions
$$
y_{1, AL}=\ell_{1, AL}^\alpha=(\xi \bar \ell+a(1-\xi)\bar \ell)^\alpha, \ \ and  \ \   y_{2,AL}=\ell_{2,AL}^\alpha=(1-a)^\alpha(1-\xi)^\alpha \bar \ell^\alpha.
$$
Note that
\begin{equation}\label{yinALAL}
y_{1, AL}>y_{1}, \quad y_{2,AL}<y_{2},
\end{equation}
where we recall that $
 y_1=(\xi \bar \ell)^\alpha, y_2=(1-\xi)^\alpha\bar \ell^\alpha,
$
are the productions before the lockdown.
Observe that this change in the pattern of production affect only the equations for $b$ and $c_2$ where at the place of $y_1$ and  $y_2$ we have now $y_{1,AL}$ and $y_{2,AL}$, respectively. Then, it follows immediately that the results of the sections \ref{Sec:2sectors} and \ref{Sec:1sector}  still hold as well as the results of Proposition \ref{Prop:prelock}.
\begin{proposition}\label{confrpnps}
Consider the price dynamics in an economy without the lockdown (NL) and in an economy with a $\tilde t$-period lockdown (AL),:
$$
	p_{t,NL}=p^*+\frac{(\phi+\psi_1)a_{c_2c_1}+\phi a_{c_2 h}}{\phi \lambda_{NL}}(h_0-h^*_{NL})e^{\psi_1 t} \ \ with \ \ t\in[0,\infty].
	$$
		$$
	p_{t,AL}=p^*_{AL}+ \frac{(\phi+\psi_1)a_{c_2c_1}+\phi a_{c_2 h}}{\phi\lambda_{AL}}\left[h^*_L-h^*_{AL}+(h_{0}-h^*_L)e^{\psi_1 \tilde t}\right]e^{\psi_1(t-\tilde t)} \ \ with \ \ t\in[\tilde t,\infty]
	$$
with $\lambda_{NL}, p^*, h^*_{NL}(=h^*)$ are as in Proposition \ref{Prop:prelock}, $\lambda_{AL}$, $p^*_{AL}$, $h^*_{AL}$ as in the anologous of Proposition \ref{Prop:prelock} with $y_{1,AL}, y_{2,AL}$ instead of $y_1, y_2$, and $h^*_L$ as  in equation (\ref{eqeq:hstarLock}).
 Under the same assumptions of Proposition \ref{Prop:prelock},  and assuming in addition that either
\begin{equation}\label{assconfrpstarnps1}
a_{c_1c_2}>0, \quad a_{c_2h}>-a_{c_1c_2}, \quad a_{c_2}>-a_{c_2c_2}y_2,
\end{equation}
or
\begin{equation}\label{assconfrpstarnps2}
a_{c_1c_2}<0, \quad a_{c_2h}>-\frac{\phi+\rho}{\phi}a_{c_1c_2}, \quad a_{c_2}>-a_{c_2c_2}y_2,
\end{equation}
then we have that
\begin{equation}\label{confrpstarnps}
p^{*}_{AL}>p^*.
\end{equation}
Moreover, denote
\begin{equation}\label{eq:I}
I=(h_0-h^*_{L})\lambda_{NL}-(h_0-h^*_{NL})\lambda_{AL}.
\end{equation}
Then we have that, if the satiation dominates the substitutability effect, $a_{c_2h}<\bar a_{c_2h}$, $h_0>h^*_{L}$, and $I>0$, and the lockdown is sufficiently short $\tilde t< \underline{\tilde t}$ then
\begin{equation}\label{eq:dpd2}
p_{ \tilde t,AL}<p_{\tilde t,NL},
\end{equation}
where $\bar a_{c_2h}\equiv -\frac{\phi+\psi_1}{\phi} a_{c_2c_1},  \underline{\tilde t}=\frac{\ln((p^*_{AL}-p^*)\phi\lambda_{NL}\lambda_{AL})-\ln(-((\phi+\psi_1)a_{c_2c_1}+\phi a_{c_2h})I)}{\psi_1}.$
\end{proposition}
\begin{proof}
First we prove \eqref{confrpstarnps}.
We think of $p^*$ as a function in two variables,
\[p^*_{AL}(z_1,z_2)=\frac{a_{c_2}+a_{c_2c_2}z_2+\left(a_{c_1c_2}+a_{c_h h}\right)h^*_{AL}(z_1)}{\lambda_{AL}(z_1,z_2)}, \quad z_1\in[y_1, y_{1, AL}], z_2\in[y_{2,AL}, y_2],
\]
where
$$
h^*_{AL}(z_1)=\frac{\phi(\psi_1-r)}{(\phi+r)\psi_1}\left[rb_0+z_1+\frac{r(\phi+\psi_1)}{\phi (\psi_1-r)}h_0\right]
$$
$$
\lambda_{AL}(z_1,z_2)=m_0(z_2)+m_1\left[rb_0+z_1+\frac{r(\phi+\psi_1)}{\phi(\psi_1-r)}h_0\right],
$$
$$
m_0(z_2)=a_{c_1}+\frac{\phi a_h+[(\phi+\rho)a_{c_1c_2}+\phi a_{c_2h}]z_2}{\phi+\rho}.
$$
In particular,
$p^*_{AL}(y_{1,AL},y_{2,AL})=p^*_{AL},\,\,\textrm{and}\,\,p^*_{AL}(y_{1},y_{2})=p^*.$
Relation \eqref{confrpstarnps} can be rewritten as
$p^*_{AL}(y_{1,AL},y_{2,AL})>p^*_{AL}(y_{1},y_{2}).$
Thus, if we prove
\begin{equation}\label{eq:derivative_z1}
    \frac{\partial p^*_{AL}}{\partial z_2}(y_{1,AL},z_2)<0
\end{equation}
and
\begin{equation}\label{eq:derivative_z2}
    \frac{\partial p^*_{AL}}{\partial z_1}(z_1,y_{2})>0
\end{equation}
then we get
\begin{equation}\label{catenader}
p^*_{AL}(y_{1,AL},y_{2,AL})\underbrace{\geq}_{\eqref{eq:derivative_z1}}p^*_{AL}(y_{1,AL},y_{2})\underbrace{\geq}_{\eqref{eq:derivative_z2}}p^*_{AL}(y_{1},y_{2}).
\end{equation}
Note that
$
\lambda_{AL}(y_{1,AL}, z_2)> 0$ for all $z_2 \in [y_{2,AL}, y_2].$
Indeed
\begin{equation}\label{signlambda}
\lambda_{AL}(y_{1,AL}, z_2)=m_0(z_2)+m_1\left[rb_0+y_{1,AL}+\frac{r(\phi+\psi_1)}{\phi(\psi_1-r)}h_0\right]>0
\end{equation}
if and only if
\begin{equation}\label{b0upper}
b_0<-\frac{m_0(z_2)}{rm_1}-\frac{y_{1,AL}}{r}+\frac{\phi+\psi_1}{\phi(r-\psi_1)}h_0 \quad \forall z_2 \in [y_{2,AL}, y_2].
\end{equation}
Now note that, if $a_{c_1c_2}>0$, we have
$
-\frac{\phi+\rho}{\phi}a_{c_1c_2}<-a_{c_1c_2},
$
hence $a_{c_1c_2}+a_{c_2h}>0$ implies
$
a_{c_2h}>-\frac{\phi+\rho}{\phi}a_{c_1c_2}.
$
Then by either assumption \eqref{assconfrpstarnps1} or \eqref{assconfrpstarnps2}, we have
\begin{equation}\label{eq:signphirho}
(\phi+\rho)a_{c_1c_2}+\phi a_{c_2h}>0,
\end{equation}
and then
$$
m_0(z_2)=a_{c_1}+\frac{\phi a_h+\overbrace{[(\phi+\rho)a_{c_1c_2}+\phi a_{c_2h}]}^{>0}z_2}{\phi+\rho}.
$$
In addition, we have that
$
m_0(y_{2,AL})\leq m_0(z_2)$ for all $z_2 \in [y_{2,AL}, y_2]
$
and since $m_1<0$, we have
$
-\frac{m_0(y_{2,AL})}{rm_1}\leq -\frac{m_0(z_2)}{rm_1}$ for all $z_2\in [y_{2,AL}. y_2]
$
We will denote by $\tilde m_0$ and $\underline{\tilde b}_0$ the constants introduced in Proposition \ref{Prop:prelock} with $y_{1,AL}, y_{2,AL}$ in the place of $y_1, y_2$ respectively.  Since by the definition of $m_0(z_2)$ and $\tilde m_{0}$, it is easy to see that
$
m_0(y_{2,AL})=\tilde m_{0},
$
we conclude that assumption
$
b_0< -\frac{\tilde m_{0}}{rm_1}+\underline{\tilde b}_{0}
$
implies \eqref{b0upper}. Moreover, using the assumption $b_0> \underline{\tilde b}_{0}$ we have
\begin{equation}\label{signhstar}
h^*_{AL}(y_{1,AL})=h^*_{AL}>0.
\end{equation}
Then by \eqref{eq:signphirho} we have
$$
\frac{\partial \lambda_{AL}(y_{1,AL},z_2)}{\partial z_2}=\frac{\partial m_0(z_2)}{\partial z_2}=\frac{(\phi+\rho)a_{c_1c_2}+\phi a_{c_2h}}{\phi+\rho}>0
$$
and
$$
\frac{\partial p^*_{AL}(y_{1,AL},z_2)}{\partial z_2}=\frac{a_{c_2c_2}\lambda_{AL}(y_{1,AL},z_2)-(a_{c_2}+a_{c_2c_2} z_2+(a_{c_1c_2}+a_{c_2h})h^*_{AL})\frac{\partial \lambda_{AL}(y_{1,AL}, z_2)}{\partial z_2}}{\lambda_{AL}(y_{1,AL}, z_2)^2}.
$$
Note that the assumption $a_{c_2}+a_{c_2c_2}y_2>0$ implies
\begin{equation}\label{assconfr}
a_{c_2}+a_{c_2c_2}y>0, \quad \forall y \in [y_{2,AL}, y_2],
\end{equation}
since $a_{c_2c_2}<0$.
Now note that, if $a_{c_1c_2}<0$, we have
$
-\frac{\phi+\rho}{\phi}a_{c_1c_2}>-a_{c_1c_2},
$
hence $a_{c_2h}>-\frac{\phi+\rho}{\phi}a_{c_1c_2}$ implies
$
a_{c_2h}>-a_{c_1c_2}.
$
Then,  by either assumption \eqref{assconfrpstarnps1} or assumption \eqref{assconfrpstarnps2} and by \eqref{assconfr}, \eqref{signlambda} and \eqref{signhstar},  we get \eqref{eq:derivative_z1}.
Now we write
\begin{equation}\label{pstary}
p^*_{AL}(z_1, y_2)=\frac{a_{c_2}+a_{c_2c_2}y_{2}+(a_{c_1c_2}+a_{c_2 h})\frac{\phi (\psi_1-r)}{(\phi +r)\psi_1}F(z_1)}{m_0+m_1 F(z_1)} \quad z_1 \in [y_1, y_{1,AL}],
\end{equation}
where
$
F(z_1)=rb_0+z_1+\frac{r(\phi+\psi_1)}{\phi(\psi_1-r)}h_0.
$
Note that $\frac{\partial F(z_1)}{\partial z_1}=1$. Then we have
$$
\frac{\partial p^*_{AL}(z_1, y_2)}{\partial z_1}=\frac{(a_{c_1c_2}+a_{c_2h})\frac{\phi (\psi_1-r)}{(\phi+r)\psi_1}m_0-m_1( a_{c_2}+a_{c_2c_2}y_{2})}{\left(m_0+m_1F(z_1)\right)^2}.
$$
By assumption \eqref{assconfrpstarnps1} or assumption \eqref{assconfrpstarnps2} and recalling that $m_0>0, m_1<0$ and $\psi_1<0$, we have
$$
\underbrace{(a_{c_1c_2}+a_{c_2h})}_{>0}\underbrace{\frac{\phi (\psi_1-r)}{(\phi+r)\psi_1}}_{>0}\underbrace{m_0}_{>0}\underbrace{-m_1( a_{c_2}+a_{c_2c_2}y_{2})}_{>0}
$$
and we get \eqref{eq:derivative_z2}.
By \eqref{eq:derivative_z1} and \eqref{eq:derivative_z2}, as shown in \eqref{catenader}, we get \eqref{confrpstarnps}.

Now we prove \eqref{eq:dpd2}.
Note that since $h^*_L=h^*_{AL}$, we have
$
h^*_{NL}=h^*<h^*_{L}=h^*_{AL}
$
and
		$$
	p_{\tilde t,AL}=p^*_{AL}+ \frac{(\phi+\psi_1)a_{c_2c_1}+\phi a_{c_2 h}}{\phi\lambda_{AL}}(h_{0}-h^*_{AL})e^{\psi_1 \tilde t}.
	$$
For convenience of notation denote
$
SSE=(\phi+\psi_1)a_{c_2c_1}+\phi a_{c_2 h}.
$
Recall that we denote by $\tilde m_0$ and $\underline{\tilde b}_0$ the constants introduced in Proposition \ref{Prop:prelock} with $y_{1,AL}, y_{2,AL}$ in the place of $y_1, y_2$ respectively. Note that by \eqref{yinALAL} the assumption $
b_0< -\frac{\tilde m_{0}}{rm_1}+\underline{\tilde b}_{0}$ implies $
b_0< -\frac{m_{0}}{rm_1}+\underline{b}_{0}
$, so that by Proposition \ref{Prop:prelock}, we have $\lambda_{NL}>0$. Moreover by the same assumption and  the analogous of Proposition \ref{Prop:prelock},  we have $\lambda_{AL}>0$.
Then if the satiation dominates the substitutability effect, $a_{c_2 h}< \overline a_{c_2h}$, that is $SSE<0$, and  if $h_0 >h^*_{L}$ (which implies $h_0>h^*_{NL}$), we have
\begin{eqnarray*}
p_{\tilde t, AL}-p_{\tilde t, NL}&=&p^*_{AL}-p^*+ \frac{SSE(h_{0}-h^*_{L})e^{\psi_1 \tilde t}}{\phi\lambda_{AL}}-\frac{SSE(h_0-h^*_{NL})e^{\psi_1 \tilde t}}{\phi \lambda_{NL}}\\ &=& \underbrace{p^*_{AL}-p^*}_{>0}+\frac{\overbrace{SSE}^{<0}[\overbrace{(h_0-h^*_{L}}^{>0})\overbrace{\lambda_{NL}}^{>0}-\overbrace{(h_0-h^*_{NL})}^{>0}\overbrace{\lambda_{AL}}^{>0}]e^{\psi_1\tilde t}}{\underbrace{\phi \lambda_{NL} \lambda_{AL}}_{>0}}.
\end{eqnarray*}
We define $I$ as in \eqref{eq:I}, assume that $I>0$ and  conclude that
$$
p_{\tilde t, AL}-p_{\tilde t, NL}=p^*_{AL}-p^*+\frac{SSE \cdot I}{\phi \lambda_{NL}\lambda_{AL}}e^{\psi_1 \tilde t}<0
$$
 if and only if
 \begin{equation}\label{eq:tildetdpd}
 \tilde t<\frac{\ln\left(-\frac{(p^*_{AL}-p^*)\phi \lambda_{NL}\lambda_{AL}}{SSE \cdot I}\right)}{\psi_1}=\frac{\ln((p^*_{AL}-p^*)\phi\lambda_{NL}\lambda_{AL})-\ln(-SSE \cdot I)}{\psi_1}
 \end{equation}
 which entails \eqref{eq:dpd2}.
\end{proof}

\subsection*{Extension: Anticipated lockdown duration}

In Lemma \ref{lem:ex1} we solve, using the maximum principle, the infinite horizon maximization problem \eqref{eq:first_step_problem}, with $t\in[T,\infty)$. Then in Lemma \ref{lem:ex2}, we solve problem \eqref{DPP_1} as a finite horizon optimization problem where the term $v(T,b(T), h(T))$ is the terminal cost.
In the following the same assumptions of Section \ref{Section:LockdownQ} hold and to avoid cumbersome notation, we write
$$
v(0,b,h)=:v(b,h).
$$

\begin{lemma}\label{lem:ex1}
Suppose that both sectors are active with production $
y_{1,NL}=(\xi \bar l)^\alpha$, and $y_{2,NL}=[(1-\xi)\bar l]^\alpha$.
Then, a given state-control quadruple $((c_{1})_{TS,AL}, (c_{2})_{TS,AL}, h_{TS,AL}, b_{TS,AL})$ is optimal for the infinite horizon maximization problem after lockdown \eqref{eq:first_step_problem}  if and only if it is a solution of the following system
\begin{eqnarray}
&& u_{c_{1_{TS,AL}}}+\mu_{TS,AL}\phi-\lambda_{TS,AL}=0 \label{foc:c1e1_1} \\
&&u_{c_{2_{TS,AL}}}-p\lambda_{TS,AL}=0  \label{foc:c2e1_1} \\
&& \dot\mu_{TS,AL} =(\phi+\rho)\mu_{TS,AL}-u_{h_{TS,AL}},\quad t\in[T,\infty)\label{foc:he1_1}\\
&& \dot \lambda_{TS,AL} =0, \quad t\in[T,\infty)\label{foc:be1_1} \\
&& \lim_{t \to + \infty} b_{TS,AL}\cdot\lambda_{TS,AL} e^{-\rho t}=0 \label{tvc:1e1_1} \\
&& \lim_{t \to + \infty} h_{TS,AL}\cdot \mu_{TS,AL} e^{-\rho t}=0\label{tvc:2e1_1}\\
&& \dot {h}_{TS,AL}=\phi(c_{1_{TS,AL}} -  h_{TS,AL}),\quad t\in(T,\infty)  \label{statehe2_1}\\
&&\dot {b}_{TS,AL}+c_{1_{TS,AL}}=r b_{TS,AL}+ y_1,\quad t\in(T,\infty)  \label{statebe2_1}\\
&& c_{2_{TS,AL}}=y_2.\label{mcc:c2_1}\\
&& h_{TS,AL}(T)=h_{TS,L}(T),\label{ic:h_1}\\
&& b_{TS,AL}(T)= b_{TS,L}(T).\label{ic:b_1}
\end{eqnarray}
Note that $(b_{TS,L}(T), h_{TS,L}(T))$ are taken as exogenously given constants at this stage of the analysis.
\end{lemma}

\begin{lemma}\label{lem:ex2}
Consider now an economy with only sector 1 being active and producing $
y_{1, L}=A\{[\xi+a(1-\xi)]\bar\ell\}^\alpha$ with $a\in(0,1).
$ Then, a given state-control quadruple $(c_{1_{TS,L}},0, h_{1_{TS,L}}, b_{1_{TS,L}})$ is optimal for the problem \eqref{DPP_1} if and only if it is solution of the following
system
\begin{eqnarray}
&& u_{c_{1_{TS,L}}}+\mu_{TS,L}\phi-\lambda_{TS,L}=0 \label{foc:c1e1} \\
&& \dot\mu_{TS,L} =(\phi+\rho)\mu-u_{h_{TS,L}} \,\,t\in[0,T)\label{foc:he1}\\
&& \lambda_{TS,L}=v_b(b_{TS,L}(T), h_{TS,L}(T)) \label{tvc:1e1} \\
&& \mu_{TS,L}=v_h(b_{TS,L}(T), h_{TS,L}(T))\label{tvc:2e1}\\
&& \dot {h}_{TS,L}=\phi(c_{1_{TS,L}} - h_{TS,L})\,\,t\in(0,T]\label{statehe1} \\
&&\dot b_{TS,L}+c_{1_{TS,L}}=rb_{TS,L}+ y_{1,L}\,\,t\in(0,T] \label{statebe1}\\
&&b_{TS,L}(0)=b_0\\
&&h_{TS,L}(0)=h_0\label{h0},
\end{eqnarray}
where $v_b(b_{TS,L}(T), h_{TS,L}(T))$,$v_h(b_{TS,L}(T),h_{TS,L}(T))$ are treated as exogenously given constants at this stage of the analysis.
\end{lemma}

\begin{proposition}\label{prop:solutiontwostage}
Let $\psi_1, \psi_2$ be defined in \eqref{eq:eigenvalues}. Denote $a=r+\phi(1-a_{c_1h}),  b=1-a_{c_1h}^2, d=\phi(a_{c_1h}-1), c=\phi^2$. The optimal state-control quadruple $({c_1}_{TS,L}, 0, h_{TS,L}, b_{TS,L})$ for the infinite horizon maximization problem after lockdown, solution of \eqref{foc:c1e1}-\eqref{h0},   satisfies
\begin{align}
&c_{1_{TS,L}}(t)=a_{c_1}+a_{c_1h} [h^*_{TS, L}+(h_0-h^*_{TS, L}-D)e^{\psi_1 t}+De^{\psi_2 t}]+\nonumber\\
&\hspace{4cm}+\phi[\mu^*_{TS, L}-(h_0-h^*_{TS, L}-D) \frac{b}{a-\psi_1}e^{\psi_1 t}-D\frac{b}{a-\psi_2}e^{\psi_2 t}]-v_b(b(\tilde t), h(\tilde t))\nonumber\\
&h_{TS,L}(t)=h^*_{TS, L}+(h_0-h^*_{TS, L}-D)e^{\psi_1 t}+De^{\psi_2 t} \label{eqn:solution2sth}\\
 &b_{TS,L}(t)=e^{rt}b_0+A(e^{rt}-1)+B(e^{\psi_1 t}-e^{rt})+C(e^{\psi_2 t}-e^{rt})\label{eqn:solution2stb}
\end{align}
where
$$
   h^*_{TS,L}=\frac{\phi \lambda_{TS,L} -\phi a_{c_1}-c\mu^*_{TS,L}}{d} ,
$$ $\lambda_{TS, L}$ is given by \eqref{tvc:1e1} and
$$
\mu^*_{TS, L}=\frac{a_h d+a_{c_1 h}a_{c_1}d-a_{c_1h}\lambda_{TS,L} d-b\phi \lambda+b\phi a_{c_1}}{ad-bc},
$$
$$
	    D=\left(v_h(b_{TS,L}(\tilde t), h_{TS,L}(\tilde t))-\mu^*_{TS,L}+(h_0-h^*_{TS,L})\frac{b}{a-\psi_1}e^{\psi_1 \tilde t}\right)\frac{(a-\psi_1)(a-\psi_2)}{be^{\psi_1\tilde t}(a-\psi_2)-be^{\psi_2 \tilde t}(a-\psi_1)},
$$
$$
		    A:=\frac{y_{1,L}-a_{c_1}-a_{c_1h} h^*_{TS,L}-\phi \mu^*_{L,TS}+v_b(b_{TS,L}(\tilde t), h_{TS,L}(\tilde t))}{r},
$$
$$
		    B:=\frac{\phi(h_0-h^*_{TS,L}-D)\frac{b}{a-\psi_1}-a_{c_1h}(h_0-h^*_{TS,L}-D)}{\psi_1-r},
$$
$$
		    	C:=\frac{\phi D \frac{b}{a-\psi_2}-a_{c_1h}D}{\psi_2-r},
$$
For $t>T$, the optimal state-control quadruple $({c_1}_{TS,AL}, {c_2}_{TS,AL}, h_{TS,AL}, b_{TS,AL})$, solution to \eqref{foc:c1e1_1}-\eqref{ic:b_1}, satisfies
\begin{align*}
& c_{1_{TS,AL}}(t)=h^*_{TS,AL}+\frac{\phi+\psi_1}{\phi}(h_{TS,L}(T)-h^*_{TS,AL})e^{\psi_1(t-T)}\\
  &h_{TS,AL}(t)= h^*_{TS,AL}+(h_{TS,L}(T)- h^*_{TS,AL})e^{\psi_1 (t-T)}\\
 &b_{TS,AL}(t)=b^*_{TS,AL}+\frac{\psi+\psi_1}{\phi(r-\psi_1)}[h_{TS,AL}(t)-h^*_{TS,AL}]\\
 &c_{2_{TS,AL}}(t)\equiv y_2
\end{align*}
where
$$
    h^*_{TS,AL}=  \frac{\phi(\psi_1-r)}{(\phi+r)\psi_1}\left[rb_{TS,L}(T)+y_1+\frac{r(\phi+\psi_1)}{\phi(\psi_1-r)}h_{TS,L}(T)\right]
$$
and $b^*_{TS,AL}$ is defined in \eqref{eq:bstar}.
\end{proposition}

It is possible to find the following explicit formula for the value function by some computations done in MATLAB.
\begin{proposition}\label{prop:valuefunction_explicit}
The value function is given by the following formula:
\begin{align}\label{eqn:v}
&v(b,h)\nonumber =(A_0+A_bb+A_h h+A_{2,b}b^2+A_{2,h}h^2+A_{b,h}bh)+\\
&+\frac{1}{\rho-2\psi_1}(\tilde A_0+\tilde A_bb+\tilde A_h h+\tilde A_{2,b}b^2+\tilde A_{2,h}h^2+\tilde A_{b,h}bh)\\
&+\frac{1}{\rho-\psi_1}(A^*_0+ A^*_bb+A^*_hh+A^*_{2,b}b^2+A^*_{2,h}h^2+A^*_{b,h}bh)
\end{align}
where the coefficients $A_0,A_b,A_h,A_{2,b},A_{2,h},A_{b,h},\tilde A_0,\tilde A_b,\tilde A_h,\tilde A_{2,b},\tilde A_{2,h},\tilde A_{b,h},A^*_0,A^*_b,A^*_h,A^*_{2,b}$\\ $A^*_{2,h},A^*_{b,h}$ are calculated using MATLAB.

\end{proposition}

		Finally, note that in Proposition \ref{prop:solutiontwostage} the constants $h^*_{TS,L}, \mu^*_{TS,L}, A, B, C, D$ and $h^*_{TS,AL}$ depend on the ``parameters" $v_b(h_{TS,L}(T), b_{TS,L}(T))$ $v_h(h_{TS,L}(T),b_{TS,L}(T))$ and on $(b_{TS,L}(T), h_{TS,L}(T))$, respectively.
To solve the problem, we need to find their explicit values.  The procedure one could use is the following. The values of $v_h(b_{TS,L}(T), h_{TS,L}(T))$, and $v_b(b_{TS,L}(T), h_{TS,L}(T))$ are given in Proposition \ref{prop:valuefunction_explicit} in terms of $b_{TS, L}(T), h_{TS, L}(T)$. Then, $(b_{TS,L}(T), h_{TS,L}(T))$ can be found by plugging $t=T$ in \eqref{eqn:solution2stb}, coupling it with the equation \eqref{eqn:solution2sth} and plugging in the expressions for the derivatives of the value function. The result is  a linear system in the unknown $(b(T), h(T))$,  easily solvable in MATLAB.

\newpage

\section*{SUPPLEMENTARY MATERIAL (For Online Pubblication)}

The purpose of this section is to show that a lockdown leads to no economic consequences after its end if the economy is the same as the one described in the main part of the paper but without habits.

Under this assumption the first order conditions are
\begin{eqnarray}
	&& u_{c_1}(c_1,c_2)=\lambda \\
	&& u_{c_2}(c_1,c_2)=p\lambda \\
	&&  \lambda = \ constant \qquad as \ we \ assume \ r=\rho \\
	&& \lim_{t\rightarrow\infty}b\lambda e^{-\rho t}=0
\end{eqnarray}
Moreover, at the equilibrium we have that
$$c_2=y_2 \qquad and \qquad \dot b+rb=y_1-c_1$$
Assuming that the marginal utility with respect to $c_1$ is an invertible function, it follows immediately that in equilibrium the consumption of good 1 is constant: $$c_1=u^{-1}_{c_1}(\lambda;y_2)= \ constant.$$
Therefore, the good 1 market equilibrium condition is a linear ODE in the variable $b$ whose solution is
$$b_t=\left(b_0+\frac{y_1-c_1}r\right)e^{rt}-\frac{y_1-c_1}r$$
Using the TVC it follows immediately that the equilibrium path of $b$ is
$$b_t=-\frac{y_1-c_1}r \qquad for \ all \ t\geq 0$$
Based on this finding, we have that good 1 consumption before, i.e. $t=0$, and during the lockdown, i.e. $t\in(0,\tilde t]$, are respectively
\begin{eqnarray*}
	&& c^{BL}_1=rb_0+y^{BL}_1 \\
	&& c^{L}_1=rb_0+y^{L}_1
\end{eqnarray*}
Observe that since $b_0$ is the same in both expressions (this is because it is a predetermined variable), then the increase in good 1 consumption during the lockdown is fully driven by an expansion of good 1 production due to the change in labor composition.
As a consequence, the economy will immediately return after the lockdown to its pre-lockdown levels since labor readjusts to its original levels, i.e.
$$y_1^{AL}=y_1^{BL} \qquad \Rightarrow \qquad c_1^{AL}=c_1^{BL},  \ and  \ \ \lambda^{AL}=\lambda^{BL}.$$
Therefore, there will not be no change in the relative price either.

We conclude that all the effects found in the paper depends on the habits and that their introduction induces also transitional dynamics in the economy.
\end{appendices}
\end{document}